\documentclass[journal]{IEEEtran}
\normalsize
\usepackage[font={small}]{caption}
\usepackage{color}
\usepackage{graphicx}
\usepackage{epstopdf}
\usepackage{subfloat}
\usepackage{float}
\usepackage{amsthm}
\usepackage{amssymb}
\newtheorem{theorem}{Theorem}

\usepackage{epsfig}
\usepackage{amsfonts}
\usepackage{epstopdf}
\usepackage[]{footmisc}
\usepackage{mathdots}
\usepackage[]{algorithmic}
\usepackage[]{algorithm2e}
\SetKwFor{While}{while}{}{end while}%
\usepackage{listings}
\usepackage[keeplastbox]{flushend}
\usepackage{fancybox}
\usepackage{epsfig}
\usepackage{mathtools}
\usepackage{subfig}
\usepackage[]{algorithm2e}

\usepackage{amsmath}
\ifCLASSINFOpdf
\else
\fi
\newcommand{\comm}[1]{}
\newcommand{\imj}{\mathsf{j}}
\newcommand{\Zmax}{Z_{\mathrm{max}}}
\newcommand{\Zmin}{Z_{\mathrm{min}}}

\allowdisplaybreaks
\usepackage{cite}
\begin{document}
\title{FALP: Fast beam alignment in mmWave systems with low-resolution phase shifters}
\author{{\IEEEauthorblockN{Nitin Jonathan Myers, {\it Student Member, IEEE},\\ Amine Mezghani, {\it Member, IEEE}, and 
Robert W. Heath Jr., {\it Fellow, IEEE}. }}
\thanks{ N. J. Myers (nitinjmyers@utexas.edu) and R. W. Heath Jr. (rheath@utexas.edu) are with the  Wireless Networking and Communications  Group, The University of Texas at Austin, Austin,
TX 78712 USA. A. Mezghani (Amine.Mezghani@umanitoba.ca) is with the Department of Electrical and Computer Engineering, University of Manitoba, Winnipeg, MB, R3T 5V6, Canada. This material is based upon work supported by the National Science
Foundation under Grant numbers NSF-CNS-1702800, NSF-CNS-1731658, and NSF ECCS-1711702.}}
%\markboth{IEEE Transactions on Communications}%
%{Submitted paper}
\maketitle 
\begin{abstract}
Millimeter wave (mmWave) systems can enable high data rates if the link between the transmitting and receiving radios is configured properly. Fast configuration of mmWave links, however, is challenging due to the use of large antenna arrays and hardware constraints. For example, a large amount of training overhead is incurred by exhaustive search-based beam alignment in typical mmWave phased arrays. In this paper, we present a framework called FALP for Fast beam Alignment with Low-resolution Phase shifters. FALP uses an efficient set of antenna weight vectors to acquire channel measurements, and allows faster beam alignment when compared to exhaustive scan. The antenna weight vectors in FALP can be realized in ultra-low power phase shifters whose resolution can be as low as one-bit. From a compressed sensing (CS) perspective, the CS matrix designed in FALP satisfies the restricted isometry property and allows CS algorithms to exploit the fast Fourier transform. The proposed framework also establishes a new connection between channel acquisition in phased arrays and magnetic resonance imaging.
\end{abstract}
\IEEEpeerreviewmaketitle
\section{Introduction}
Millimeter wave (mmWave) communication, currently used in the IEEE 802.11ad standard and 5G, can support Gbps data rates by exploiting the large amount of bandwidth available at mmWave frequencies \cite{mmintro}. MmWave systems typically use large antenna arrays and directional beams to achieve such high data rates. The process of beam alignment, i.e., finding the best directional beam, can be challenging in mmWave hardware architectures like the phased array \cite{heathoverview}. As phased arrays have fewer radio frequency (RF) chains than the antennas, standard techniques like exhaustive search-based beam alignment can result in a substantial training overhead when applied to mmWave systems\cite{heathoverview}.  
%Although existing circuit technology can support communication at mmWave frequencies, link configuration in typical mmWave hardware can be challenging \cite{heathoverview}. For example, in a phased antenna array-based mmWave radio that is used in IEEE 802.11ad standard compliant devices \cite{11ad}, 
%\par Multiple-input multiple-output (MIMO) channel matrices at mmWave carrier frequencies have special structure that arises due to the propagation characteristics of the environment \cite{heathoverview}. For example, mmWave channels are approximately sparse in a dictionary that depends on the array geometries at the transmitter and the receiver. 
%The sparsity of mmWave MIMO channels can be exploited to perform channel reconstruction or beam alignment using sub-Nyquist channel measurements, i.e., fewer channel measurements when compared to the dimension of the channel. 
%Compressed sensing is a technique to efficiently acquire and reconstruct sparse signals from fewer projections of the signal when compared to its dimension \cite{csintro,csmag}. CS-based algorithms that exploit sparsity of mmWave channel are useful as they can reduce the training duration for beam alignment or channel estimation when compared to conventional techniques \cite{kiranchannel,cschest}.
\par Compressed sensing (CS) is a technique that allows reconstructing a sparse signal with fewer measurements when compared to the dimension of the signal \cite{csintro}. Due to the sparse nature of mmWave channels in an appropriate dictionary, CS is a promising solution for mmWave channel estimation or beam alignment with sub-Nyquist channel measurements \cite{heathoverview}. The channel measurements in CS are obtained by projecting the channel onto a lower dimensional subspace using a CS matrix \cite{csintro}. The channel is then recovered from the lower dimensional projections using optimization techniques that exploit sparsity of the channel \cite{cschest,javiCS}. The guarantees on the recovery of sparse signals and the complexity of the reconstruction algorithms, depend on the choice of the CS matrix used to obtain these projections. The restricted isometry property (RIP) \cite{candes2008restricted} is one metric that characterizes the efficiency of a CS matrix in recovering sparse signals. Unfortunately, several random CS matrices that are known to satisfy the RIP cannot be realized in phased arrays due to hardware constraints \cite{swiftlink}. To this end, prior work has used random IID phase shift-based CS matrices for sub-Nyquist mmWave channel estimation and beam alignment \cite{cschest,javiCS}. CS techniques that use the random phase shift design, however, cannot exploit fast transforms and may result in a high complexity when applied to large antenna systems. Structured CS algorithms are promising for large antenna systems as they can perform sparse recovery with a reduced computational complexity \cite{kramer}.
\par Convolutional compressed sensing (CCS) is one form of structured CS in which the signal of interest is projected onto fewer circulantly shifted versions of a base sequence \cite{convCS}. The convolutional structure of CS matrices in CCS allows sparse recovery algorithms to exploit the fast Fourier transform. In CCS of vectors, the choice of the base sequence is critical for the successful recovery of the sparse signal \cite{convCS}. Prior work has shown that optimal base sequences of a certain length exist when the size of the alphabet is sufficiently large \cite{convCS,perf_seq_exist}. For example, ideal base sequences of length $16$ exist in $\{1,\imj,-1,-\imj\}^{16}$, where $\imj=\sqrt{-1}$. Base sequences of the same length, however, do not exist in $\{1,-1\}^{16}$ \cite{perf_seq_exist}. For CCS in phased arrays, the size of the alphabet is determined by the resolution of the phase shifters. As typical mmWave phased arrays use low resolution phase shifters, applying CCS in such systems can be challenging. The  difficulty lies in finding optimal base sequences that are compatible with the hardware. In this paper, we construct efficient structured CS matrices that satisfy the RIP, and can be realized in arbitrarily large phased arrays whose resolution can be as low as one-bit.
\par We propose a novel 2D-CCS framework called FALP for convolutional compressed sensing of mmWave channels with planar phased arrays. The channel measurements in our 2D-CCS framework are obtained by projecting the channel matrix onto 2D-circulant shifts of a base matrix. Similar to standard vector CCS \cite{convCS}, the performance of 2D-CCS depends on the choice of the base matrix. As the number of hardware compatible matrices in phased arrays is exponential in the array dimensions, a brute-force approach to find the best base matrix is not practical for large arrays. In our prior work called Swift-Link \cite{swiftlink}, we used Zadoff-Chu (ZC) sequences for efficient CCS-based channel estimation in planar arrays. For linear phased arrays, the efficiency of CCS-based channel estimation with ZC sequences was studied in \cite{struct_rand_phased_array}. The 2D-CCS equivalent of the sequences in  \cite{swiftlink} and \cite{struct_rand_phased_array}, is a base matrix that is an outer product of two ZC sequences. Realizing base matrices using ZC-sequences, however, requires a phase shift resolution that is logarithmic in the number of antennas \cite{struct_rand_phased_array}. In large antenna arrays, it can be difficult to meet such a requirement as the use of high resolution phase shifters can result in a higher hardware cost and a higher power consumption. 
\par One-bit phased arrays are promising in terms of the hardware complexity, cost and power consumption \cite{wang2018hybrid,onebitlis}. The hardware associated with a one-bit phase shifter can be as simple as a combination of a switch and an inverter \cite{onebit_hw}. Both these components consume less power than a typical high resolution phase shifter. For instance, a one-bit phase shifter may require $10\, \mathrm{mW}$ while a four-bit phase shifter may need $45\, \mathrm{mW}$ \cite{onebit_hw}. The binary phase control capability in one-bit phased arrays, however, complicates the design of efficient base matrices for 2D-CCS. In this paper, we arrive at a surprising result that ideal base matrices for 2D-CCS exist for infinite array dimensions even over a binary alphabet. This result allows applying FALP to large phased arrays with one-bit phase shifters, and makes it a candidate solution for next generation wireless systems. We summarize our main contributions as follows.
\begin{itemize}
  \item We propose a compressive channel acquisition technique that acquires channel measurements with fewer 2D-circulant shifts of a base matrix. We determine the properties of base matrices that result in efficient 2D-CCS and are compatible with phased arrays.
  \item We show that perfect arrays \cite{PBA,PQA} can be used as efficient base matrices in FALP. For a given resolution of phase shifters, such arrays exist for a family of array dimensions. For other cases, we derive the sub-optimality gap of CS algorithms when non-ideal base matrices are used in our framework. 
\item We establish an equivalence between CS-based beam alignment with FALP and CS in magnetic resonance imaging (MRI) \cite{sparseMRI}. The equivalence allows direct application of k-space trajectories in MRI to the beam alignment problem. For a random trajectory, we derive the probability of successful beam alignment using zero filling-based reconstruction in MRI. We use this equivalence to show how low complexity beam alignment can be performed using a single 2D-fast Fourier transform.
%\item Using simulations, we show that the zero-filling reconstruction technique in MRI can be used for efficient beam alignment with FALP.
\item Using simulations, we show that the use of perfect arrays in 2D-CCS results in better beam alignment when compared to 2D-CCS with a randomly chosen base matrix. We also show that the proposed CS technique performs slightly better than the common random phase shift-based approach, for a significantly reduced computational complexity.
\end{itemize}
We would like to highlight that Swift-Link \cite{swiftlink} and FALP solve two independent problems. On the one hand, FALP develops efficient base matrices for 2D-CCS in low resolution phased arrays. On the other hand, Swift-Link designs trajectories to perform CS-based beam alignment that is robust to carrier frequency offset (CFO). For simplicity of exposition, we assume perfect frame timing and carrier synchronization. Nevertheless, Swift-Link's trajectory can be used in FALP for CFO robust beam alignment in low-resolution phased arrays.
\par The rest of the paper is organized as follows. In Section~\ref{sec:syschanmodel}, we describe the system and channel model in a planar phased array-based system. Section~\ref{sec:core_CCS} is the main technical section of the paper, where we explain how channel measurements are acquired in 2D-CCS and introduce the notion of base matrix. We mathematically show that perfect arrays \cite{PBA,PQA} are good candidates for ideal base matrices, and describe FALP in Section~\ref{sec:core_CCS}. In Section~\ref{sec:MRI_BA}, we explain how compressive beam alignment in FALP is analogous to CS in MRI. We use the analogy to develop a beam alignment technique that does not require any iterative optimization. Simulation results are presented in Section~\ref{sec:simulations}, before the conclusions and future work in Section~\ref{sec:concl_fw}.  
\par \textbf{Notation}$:$ $\mathbf{A}$ is a matrix, $\mathbf{a}$ is a column vector and $a, A$ denote scalars. Using this notation $\mathbf{A}^T,\mathbf{A}^{\text{c}}$ and $\mathbf{A}^{\ast} $ represent the transpose, conjugate and conjugate transpose of $\mathbf{A}$. We use $\mathrm{diag}\left(\mathbf{a}\right)$ to denote a diagonal matrix with entries  of $\mathbf{a}$ on its diagonal. The scalar $a\left[m \right]$ denotes the $m^{\mathrm{th}}$ element of $\mathbf{a}$. The $\ell_2$ norm of $\mathbf{a}$ is denoted by $\Vert \mathbf{a} \Vert_2$. The $k^{\mathrm{th}}$ row and the $\ell^{\mathrm{th}}$ column of $\mathbf{A}$ are denoted by $\mathbf{A}(k,:)$ and $\mathbf{A}(:,\ell)$. The scalar $\mathbf{A}\left(k,\ell\right)$ or $\mathbf{A}_{k,\ell}$ denotes the entry of $\mathbf{A}$ in the $k^{\mathrm{th}}$ row and the ${\ell}^{\mathrm{th}}$ column. The matrix $|\mathbf{A}|$ contains the element-wise magnitude of $\mathbf{A}$, i.e., $|\mathbf{A}|_{k, \ell}=|\mathbf{A}_{k,\ell}|$. The $\ell_1$ norm and the Frobenius norm of $\mathbf{A}$ are denoted by $\Vert \mathbf{A} \Vert_1$ and $\Vert \mathbf{A} \Vert_{\mathrm{F}}$. The inner product of two matrices $\mathbf{A}$ and $\mathbf{B}$ is defined as $\langle \mathbf{A},\mathbf{B}\rangle =\sum_{k,\ell}\mathbf{A}\left(k,\ell \right){\mathbf{B}^{\text{c}}}\left(k,\ell\right)$. We use $\mathbf{1}$ to denote an all-ones matrix and $\mathbf{I}$ to denote the identity matrix. The symbols $\odot$ and $\circledast$ are used for the Hadamard product and 2D circular convolution \cite{imageprocess}. %The matrix $\mathbf{U}_N \in \mathbb{C}^{N \times N}$ denotes the unitary Discrete Fourier Transform (DFT) matrix. The set $\mathcal{I}_N$ denotes the set of integers $\left\{ 0,1,2,...,N-1\right\}$. We use $\mathbf{e}_k$ to represent the $(k+1)^{\mathrm{th}}$ canonical basis vector.  
\section{System and channel model} \label{sec:syschanmodel}
In this section, we describe a planar phased antenna array system considered in FALP. To explain our framework, we assume a narrowband mmWave system and focus on the transmit beam alignment problem. We extend our algorithm to the wideband setting in Section \ref{sec:simulations}. 
\subsection{System model} \label{sec:sysmodel}
\begin{figure}[h]
\centering
\includegraphics[trim=1.25cm 0cm 0.8cm 0cm, width=0.48\textwidth]{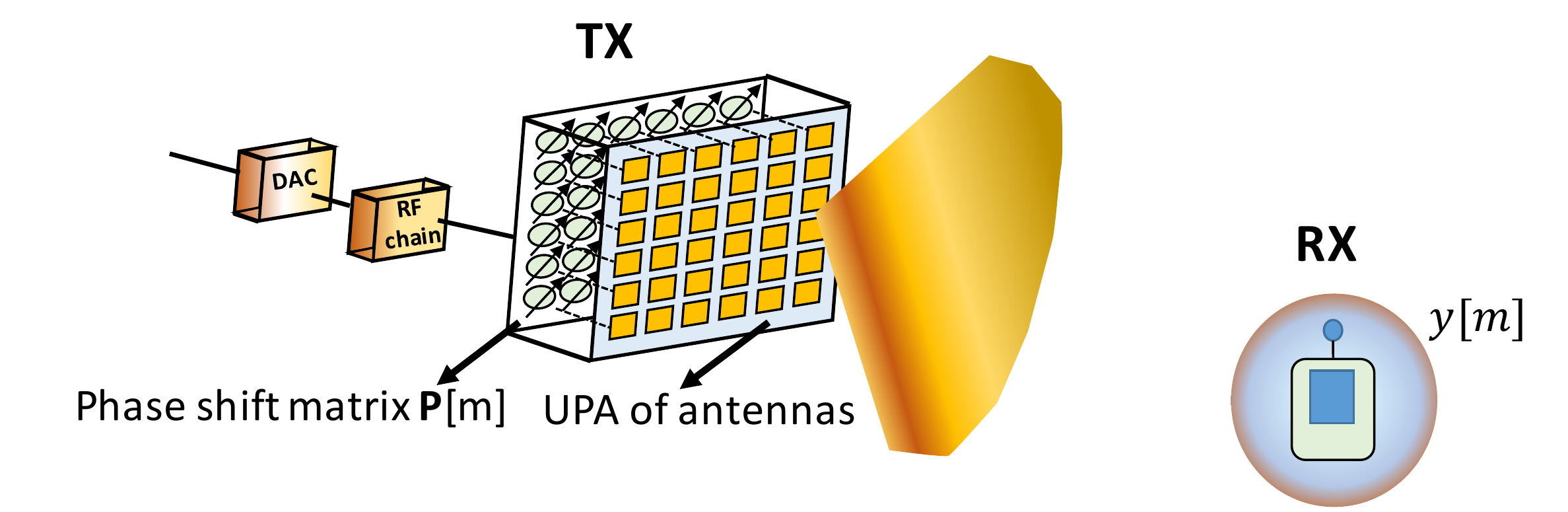}
\caption{\small Channel acquisition in a phased array system with a uniform planar array of antennas at the transmitter (TX). The channel measurements at the receiver are used to estimate the best phase shift configuration at the TX.}
\normalsize
  \label{fig:architect}
\end{figure}
We consider an analog beamforming system in which the transmitter (TX) is equipped with a uniform planar array (UPA) of antennas as shown in Fig.~\ref{fig:architect}. For ease of notation, we consider an equal number of antennas, i.e., $N$, along each of the azimuth and elevation dimensions of the UPA. Our framework can also be extended to other rectangular arrays by using array response vectors of appropriate dimensions in the formulation. The beamforming architecture at the TX uses a single radio frequency (RF) chain as shown in Fig.~\ref{fig:architect}. Each antenna element in the UPA is connected to the RF chain through a digitally controlled phase shifter.
%\comm{A phase shifter is a delay element that is used to control the phase of the mmWave RF signal transmitted from an antenna.}
By appropriately configuring the \comm{collection of}phase shifters, the TX can perform directional transmission \cite{heathwicomm}. We define the set $\mathcal{I}_{J}=\{0,1,2,\cdots,J-1\}$. The resolution of each phase shifter is assumed to be $q$-bits; the set of possible phase shifts is defined as $\mathbb{Q}_q=\left\{e^{\imj 2 \pi k / 2^q} /N : k \in \mathcal{I}_{2^q} \right\}$. As each antenna in the UPA is connected to a unique phase shifter, it is possible to configure $N^2$ phase shifters. Therefore, the phase shift matrix applied to the phased array at the TX is constrained to be an element in $\mathbb{Q}^{N \times N}_q$. The transmit beam alignment problem is to determine a phase shift matrix at the TX that maximizes the SNR at the receiver (RX).  
\par A possible approach to perform beam alignment is to estimate a reasonable approximation of the channel and use it to configure the phased array. For simplicity of exposition, we assume a single antenna at the RX and focus on the transmit beam alignment problem. Our framework can be extended to settings with UPAs at both the TX and the RX, by using fourth order tensors to model the channel. We index the antenna element in the $i^{\mathrm{th}}$ row and the $j^{\mathrm{th}}$ column of the transmit array as $(i,j)$. For an $N \times N$ UPA, $i \in \mathcal{I}_N$ and $j \in \mathcal{I}_N$. Let $\mathbf{H} \in \mathbb{C}^{N \times N}$ be the channel matrix between the UPA at the TX and the receive antenna. Specifically, $\mathbf{H}(i,j)$ represents the channel coefficient between the $(i,j)^{\mathrm{th}}$ antenna in the UPA and the antenna at the RX. The TX uses different phase shift configurations across multiple training slots for the RX to obtain channel measurements. 
\par In the $m^{\mathrm{th}}$ training slot, the TX applies the phase shift matrix $\mathbf{P}[m] \in \mathbb{Q}_q^{N \times N}$ to its phased array, and the RX acquires the channel measurement $y[m]$. We use $M$ to denote the total number of channel measurements acquired by the RX. In this paper, we assume perfect frame timing and carrier synchronization. Our assumption is valid in cellular scenarios where synchronization is performed using separate control channels. With the perfect synchronization assumption, the $m^{\mathrm{th}}$ channel measurement is
\begin{equation}
y[m]=\langle \mathbf{H}, \mathbf{P}[m] \rangle + v[m],
\label{eq:sysmod}
\end{equation}
where $v[m]\sim \mathcal{N}_{\mathrm{c}}\left(0,\sigma^2\right)$ is additive white Gaussian noise. As the measurement in \eqref{eq:sysmod} is a scalar projection of $\mathbf{H}$, estimating a generic $N\times N$ channel matrix requires $M=N^2$ channel measurements. Exhaustive beam search is one such approach that obtains the projections of $\mathbf{H}$ on all the $N^2$ elements of the 2D-discrete Fourier transform (2D-DFT) dictionary \cite{heathoverview}. Such a solution, however, does not scale well with the array dimensions. In this paper, we propose a novel set of phase shift matrices, $\{\mathbf{P}[m]\}^{M-1}_{m=0}$, for compressive channel acquisition. We prove that a good approximation of mmWave channels can be obtained from $M=\mathcal{O}(\mathrm{log} N)$ channel measurements that are acquired using the proposed set. We also show that CS algorithms that use the proposed design have a lower computational complexity than those that use the common random phase shift-based design \cite{cschest,javiCS}.
%To overcome this issue, CS algorithms that exploit the structure in mmWave channels to estimate them with $M<N^2$ measurements must be developed. Furthermore, for scalability of the beam alignment algorithm to a large number of users, it is desirable to perform non-adaptive channel estimation, i.e., the RX does not send any feedback to the TX while acquiring the $M$ channel measurements. Designing $M<N^2$ phase shift matrices for efficient CS is a non-convex problem due to hardware constraints. For instance, the entries of each phase shift matrix must come from a finite alphabet with constant envelope, i.e., $\mathbb{Q}_q$.
%\comm{ In some cases, the phase shift matrix may also be constrained to be rank $1$, i.e., $\mathbf{P}=N\mathbf{p}_e\mathbf{p}^T_a$ for some $\mathbf{p}_e \in \mathbb{Q}_q^N$ and $\mathbf{p}_a \in \mathbb{Q}_q^N$. The codebook that constitutes all such rank $1$ phase shift matrices is commonly known as kronecker product codebook.} 
%In this paper, we design $M=\mathcal{O}(\mathrm{log}N)$ phase shift matrices in $\mathbb{Q}^{N \times N}_q$ such that it is possible to perform channel estimation or beam alignment using projections of the mmWave channel on the $M$ matrices. 
\subsection{Channel model}
We consider a geometric-ray-based model for the channel matrix $\mathbf{H}$\cite{heathwicomm}. Let $\gamma_{k}$, $\theta_{e,k}$ and $\theta_{a,k}$ denote the complex ray gain, elevation angle-of-departure and azimuth angle-of-departure of the $k^{\mathrm{th}}$ ray. We define the beamspace angles $\omega_{a,k}=\pi \, \mathrm{sin}\, \theta_{e,k} \mathrm{sin}\, \theta_{a,k}$ and $\omega_{e,k}=\pi \, \mathrm{sin}\, \theta_{e,k} \mathrm{cos}\, \theta_{a,k} $. We define the  Vandermonde vector $\mathbf{a}\left(\omega \right) \in \mathbb{C}^{N\times 1}$ as 
\begin{equation}
\mathbf{a}\left(\omega \right)=\left[1\,, e^{\imj \omega}\,, e^{ \imj 2\omega}\,, \cdots\,, e^{\imj (N-1)\omega}\right]^{T}.
\end{equation}
The wireless channel for a half wavelength spaced UPA in the baseband is given by
\begin{equation}
\mathbf{H}=\sum_{k=1}^{K}\gamma_{k}\mathbf{a} \left(\omega_{e,k}\right)\mathbf{a}^{T}\left(\omega_{a,k}\right).
\label{eq:nbchannel}
\end{equation}
%The channel is comprised of $K$ rays that are characterized by their gain and the angle of departure. Each of the $K$ components in \eqref{eq:nbchannel} is an outer product of the array response vectors along the elevation and the azimuth directions. 
As large antenna arrays are used in typical mmWave settings, the dimension of the channel, i.e., $N^2$, can be large in mmWave systems when compared to conventional lower frequency systems.
%The dimension of the channel, i.e., $N^2$, can be large in typical mmWave systems due to the use of large antenna arrays.
\par Channel matrices at mmWave are sparse in a well chosen dictionary, because of the propagation characteristics of the environment \cite{heathoverview}. For UPAs, the 2D-DFT basis is often chosen for a sparse representation of $\mathbf{H}$ \cite{beam2DDFTsparse}. We use $\mathbf{U}_N$ to denote the standard unitary DFT matrix of size $N \times N$. Let $\mathbf{X} \in \mathbb{C}^{N \times N}$ denote the inverse 2D-DFT of $\mathbf{H}$, such that   
\begin{equation}
\mathbf{H}= \mathbf{U}_{N} \mathbf{X} \mathbf{U}_{N}.
\label{eq:mimoangledom}
\end{equation}   
The unitary nature of the DFT implies that $\mathbf{X}=\mathbf{U}^{\ast}_{N} \mathbf{H} \mathbf{U}^{\ast}_{N}$. The matrix $\mathbf{X}$ is called the beamspace channel as it contains the received measurements when different directional 2D-DFT beams are used at the TX \cite{beam2DDFTsparse}. The sparsity of the mmWave channel in the angle domain translates to the sparsity of the beamspace channel matrix $\mathbf{X}$. As the beamspace angles-of-departure (AoD) in the channel may not align exactly with those corresponding to the DFT dictionary, there can be leakage effects in the 2D-DFT representation \cite{heathoverview}. Therefore, the matrix $\mathbf{X}$ is approximately sparse. In such a case, dictionaries that use a finer AoD domain representation can be used for a sparser representation of $\mathbf{H}$ \cite{cschest}. Using such a dictionary, however, increases the dimensionality of the CS problem. For our analysis, we consider $\mathbf{X}$ to be perfectly sparse, while our simulation results are for the realistic case where $\mathbf{X}$ is approximately sparse.
\begin{figure*}[h!]
\centering
\includegraphics[trim=2.5cm 0cm 3cm 0cm, width=0.9 \textwidth]{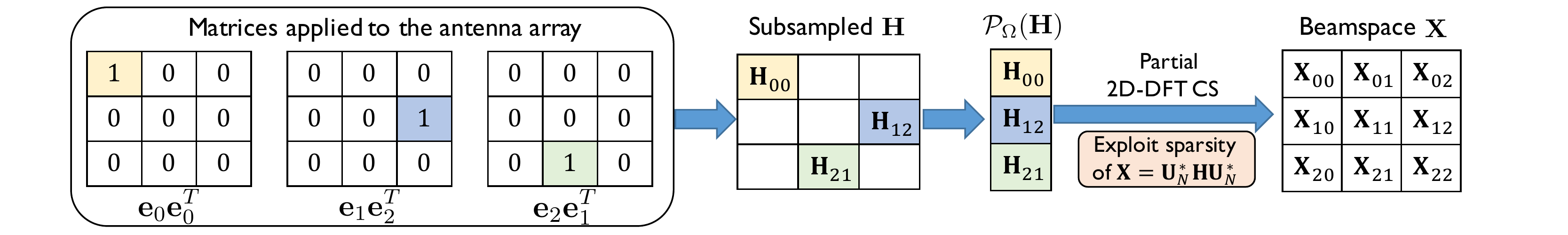}
\caption{\small The matrices used to acquire channel projections in partial 2D-DFT CS are 2D-circulant shifts of $\mathbf{e}_0\mathbf{e}^T_0$. Here, the set of circulant shifts is $\Omega=\{(0,0), (1,2),(2,1)\}$. The vector $\mathcal{P}_{\Omega}(\mathbf{H})$ is a subsampled version of $\mathbf{H}$ at the locations in $\Omega$. Partial 2D-DFT CS estimates the sparse matrix $\mathbf{X}$  from its subsampled 2D-DFT, i.e.,  $\mathcal{P}_{\Omega}(\mathbf{H})$.
}
\normalsize
  \label{fig:p2ddft_illus}
\end{figure*}
\section{Convolutional CS in planar arrays}\label{sec:core_CCS}
\par In this section, we explain the main motivation for 2D-CCS, and describe the notions of the base matrix and the sub-sampling set in 2D-CCS. Then, we identify the conditions on the base matrix that minimize the channel reconstruction error with 2D-CCS. Finally, we show that perfect arrays over small alphabets \cite{PBA,PQA} can be used as base matrices, for efficient 2D-CCS in low resolution phased arrays.
\subsection{Motivation for 2D-CCS}
\par A possible approach to acquire measurements in CS is to obtain fewer projections of the sparse signal in an appropriate basis \cite{csintro}. For example, CS can efficiently recover a sparse matrix from its subsampled 2D-DFT \cite{csintro}; the reconstruction problem in this case is known as a partial 2D-DFT CS problem \cite{rauhut2010compressive}. Partial 2D-DFT CS has a lower complexity and may achieve better signal reconstruction, when compared to other CS techniques \cite{partial2DDFTCS,ZCglobecom}. In the context of mmWave channels, partial 2D-DFT CS can estimate the sparse matrix $\mathbf{X}$ from fewer samples of its 2D-DFT, i.e., $\mathbf{H}$. An illustration of the reconstruction is shown in Fig. \ref{fig:p2ddft_illus}. The direct application of partial 2D-DFT CS in mmWave phased arrays, however, is challenging. The difficulty arises because $\mathbf{H}$ cannot be directly subsampled using phased arrays, as required by partial 2D-DFT CS. For instance, acquiring $\mathbf{H}(0,0)=\langle \mathbf{H}, \mathbf{e}_0\mathbf{e}^T_0 \rangle$ in a single training slot requires the application of $\mathbf{e}_0\mathbf{e}^T_0$ to the antenna array. The matrix $\mathbf{e}_0\mathbf{e}^T_0$, however, does not belong to the feasible set, i.e., $\mathbb{Q}^{N \times N}_q$. Although introducing switches after each phase shifter can help realize $\mathbf{e}_0\mathbf{e}^T_0$, the SNR in the resulting channel measurement can be poor. This is because $\mathbf{e}_0\mathbf{e}^T_0$ uses a single transmit antenna and per-antenna power constraints limit the power that can be transmitted from an antenna. In this paper, we develop a novel 2D-CCS technique that overcomes these practical challenges, and has all the advantages of partial 2D-DFT CS.   
\par The motivation for 2D-CCS comes from the observation that the matrices used to obtain channel projections in partial 2D-DFT CS are 2D-circulant shifts of a particular matrix. It can be noticed from Fig. \ref{fig:p2ddft_illus} that the matrices used to acquire channel projections, i.e., $\mathbf{e}_0\mathbf{e}^T_0$, $\mathbf{e}_1\mathbf{e}^T_2$, and $\mathbf{e}_2\mathbf{e}^T_1$, are all 2D-circulant shifts of $\mathbf{e}_0\mathbf{e}^T_0$. %The circulant shifts associated with these projection matrices are $(0,0)$, $(1,2)$, and $(2,1)$. We use $\Omega$ to denote the set of 2D-circulant shifts used to acquire channel measurements, and $\mathcal{P}_{\Omega}$ to denote the subsampling operator that returns the entries of a matrix at the locations in $\Omega$. As shown in Fig. \ref{fig:p2ddft_illus}, partial 2D-DFT CS estimates $\mathbf{X}$ from $\mathcal{P}_{\Omega}(\mathbf{H})$, using a sparse prior on $\mathbf{X}$. 
The channel projections in our framework are acquired by applying 2D-circulant shifts of a matrix $\mathbf{P} \in \mathbb{Q}^{N \times N}_q$, instead of $\mathbf{e}_0\mathbf{e}^T_0$, to the phased array. We define $\mathbf{P}$ as the base matrix in 2D-CCS. Due to the constant modulus nature of the matrices in $ \mathbb{Q}^{N \times N}_q$, our 2D-CCS framework uses all the antennas in the phased array for compressive channel acquisition.
\par Now, we explain how compressive channel measurements are obtained in 2D-CCS. In the $m^{\mathrm{th}}$ training slot, the TX applies a $(r[m],c[m])$ 2D-circulant shift of $\mathbf{P}$ to its phased array. The matrix $\mathbf{P}[m]$ is generated by circulantly shifting  $\mathbf{P}$ by $r[m]$ units along the rows and $c[m]$ units along the columns. We define $\Omega$ as a set that contains the 2D-circulant shifts used to acquire the $M$ channel measurements, i.e., $\Omega=\{ (r[m],c[m]) \}_{m=0}^{M-1}$. In this paper, the set of circulant shifts, i.e., $\Omega$, is constructed by sampling $M$ distinct coordinates at random from $\mathcal{I}_N \times \mathcal{I}_N$. We define $\mathbf{J} \in \mathbb{R}^{N \times N}$ as a circulant delay matrix with its first row as $(0,1,0,0,..,0)$. The subsequent rows of $\mathbf{J}$ are generated by right circulantly shifting the previous row by $1$ unit. Using this notation, we define the $d$ circulant delay matrix as $\mathbf{J}_d=\mathbf{J} \cdot \mathbf{J}\cdots \mathbf{J}$ ($d$ times). %Analogous to the canonical sampling case, a channel measurement is obtained by applying a 2D-circulantly shifted version of the base matrix $\mathbf{P}$ to the phased array. In the $m^{\mathrm{th}}$ training slot, the TX applies a phase shift matrix $\mathbf{P}[m]$, that is an $(r[m],c[m])$ 2D-circulantly shifted version of $\mathbf{P}$. 
In 2D-CCS, the matrix applied to the phased array in the $m^{\mathrm{th}}$ slot is 
\begin{equation}
\mathbf{P}[m]=\mathbf{J}^T_{r[m]}\mathbf{P}\mathbf{J}_{c[m]}.
\label{eq:Pn_from_P}
\end{equation} 
An illustration of the compressive channel acquisition procedure using 2D-CCS, for a base matrix $\mathbf{P}$ and a subsampling set $\Omega=\{(0,0),(1,2),(2,1)\}$, is shown in Fig. \ref{fig:2dccs_illus}. The base matrix determines the success of 2D-CCS-based recovery. As an example, consider a 2D-CCS technique that uses $\mathbf{P}=\mathbf{1}/N$. Channel acquisition with such a matrix results in the same measurement, i.e., mean of the entries in $\mathbf{H}$, for any 2D-circulant shift. As $\mathbf{H}$ cannot be estimated just from its mean, 2D-CCS with $\mathbf{P}=\mathbf{1}/N$ fails. In Sections ~\ref{sec:ccstransformed} and ~\ref{sec:conditions_base}, we use ideas from partial 2D-DFT CS to study how the choice of $\mathbf{P}$ impacts the performance of 2D-CCS.
\begin{figure*}[h!]
\centering
\includegraphics[trim=2.5cm 0cm 3cm 0cm, width=0.9 \textwidth]{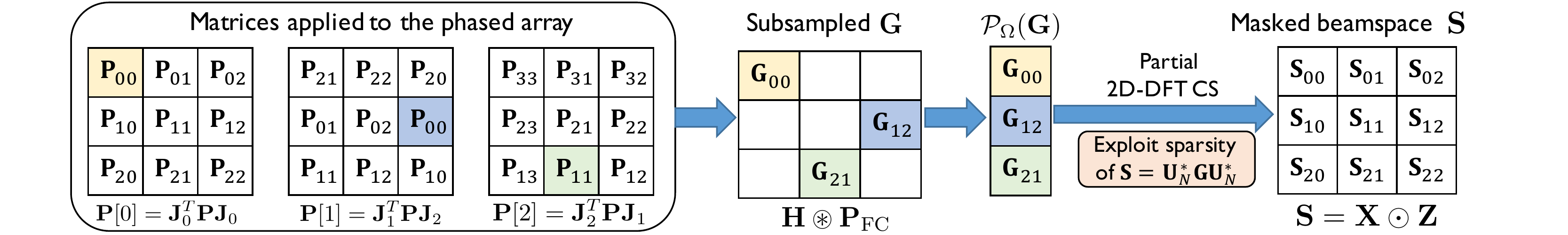}
\caption{\small Channel projections in 2D-CCS are acquired using 2D-circulant shifts of a base matrix $\mathbf{P}$. The matrix $\mathbf{P}_{\mathrm{FC}}$ is a flipped and conjugated version of $\mathbf{P}$, and $\mathbf{Z}$ is the scaled inverse 2D-DFT of $\mathbf{P}_{\mathrm{FC}}$. Partial 2D-DFT CS can be used to estimate the sparse masked beamspace matrix $\mathbf{S}$ from a sub-sampled version of its 2D-DFT, i.e., $\mathcal{P}_{\Omega}(\mathbf{G})$.}
\normalsize
  \label{fig:2dccs_illus}
\end{figure*}
\subsection{Transforming convolutional CS to partial 2D-DFT CS}\label{sec:ccstransformed}
\par We derive a compact representation of the channel measurements in 2D-CCS. In the $m^{\mathrm{th}}$ training slot, the TX applies $\mathbf{P}[m]$ to its phased array and the RX receives  
\begin{equation}
\label{eq:y_innp}
y[m]=\langle\mathbf{H},\mathbf{J}^T_{r[m]}\mathbf{P}\mathbf{J}_{c[m]} \rangle +v[m].
\end{equation}
We use $k_N$ to denote the modulo$-N$ remainder of $k$. The $m^{\mathrm{th}}$ channel measurement is then
\begin{equation}
\label{eq:y_innp_exp}
y[m]=\sum^{N-1}_{k=0}\sum^{N-1}_{\ell=0} \mathbf{H}(k,\ell) \mathbf{P}^{\text{c}}\left((k-r[m])_{N},(\ell-c[m])_{N}\right) +v[m].
\end{equation}
We define $\mathbf{P}_{\mathrm{FC}}$ as a flipped and conjugated version of $\mathbf{P}$, i.e., 
\begin{equation}
\mathbf{P}_{\mathrm{FC}} (k,\ell)= \mathbf{P}^{\text{c}}(-k_N,-\ell_N)\;\; \forall k, \ell \in \mathcal{I}_N.
\end{equation}
By the definition of 2D-circular convolution \cite{imageprocess}, it can be observed from \eqref{eq:y_innp_exp} that
\begin{equation}
\label{eq:y_circ_conv}
y[m]=(\mathbf{H}\circledast \mathbf{P}_{\mathrm{FC}})_{r[m],c[m]} +v[m].
\end{equation} 
In 2D-CCS, the RX acquires the $(r[m],c[m])$ entry of $\mathbf{H}\circledast \mathbf{P}_{\mathrm{FC}}$, when the TX applies an  $(r[m],c[m])$ 2D-circulant shift of $\mathbf{P}$ to its phased array. In $M$ training slots, the TX applies $M$ distinct 2D-circulant shifts of $\mathbf{P}$ according to the coordinates in $\Omega$. The vector of $M$ channel measurements received at the RX is then 
\begin{equation}
\label{eq:subsampled_conv}
\mathbf{y}=\mathcal{P}_{\Omega}(\mathbf{H} \circledast \mathbf{P}_{\mathrm{FC}})+\mathbf{v}.
\end{equation}    
The measurement vector in 2D-CCS is a subsampled convolution of $\mathbf{H}$ and $\mathbf{P}_{\mathrm{FC}}$.  
\par Now, we show how channel measurements in 2D-CCS can be interpreted as partial 2D-DFT measurements of a transformed beamspace. We define the convolved channel $\mathbf{G}$ as 
\begin{align}
\label{eq:modchannel}
\mathbf{G}=\mathbf{H} \circledast \mathbf{P}_{\mathrm{FC}}.
\end{align}
The spectral mask corresponding to the base matrix $\mathbf{P}$ is defined as 
\begin{equation}
\label{eq:specmask_defn}
\mathbf{Z}=N\mathbf{U}^{\ast}_N \mathbf{P}_{\mathrm{FC}} \mathbf{U}^{\ast}_N.
\end{equation}
Similar to the definition of the beamspace $\mathbf{X}$, we define the masked beamspace as  
\begin{align}
\label{eq:maskbeam_defn}
\mathbf{S}&=\mathbf{U}^{\ast}_N \mathbf{G} \mathbf{U}^{\ast}_N.
\end{align}
An interesting property of the Fourier transform is that the 2D-DFT of $\mathbf{H}\circledast \mathbf{P}_{\mathrm{FC}}$ is a scaled element-wise product of the 2D-DFTs of $\mathbf{H}$ and $\mathbf{P}_{\mathrm{FC}}$ \cite{imageprocess}. We use this property to rewrite \eqref{eq:maskbeam_defn} as
\begin{align}
\label{eq:maskbeam_hdmd1}
\mathbf{S}&=(\mathbf{U}^{\ast}_N \mathbf{H} \mathbf{U}^{\ast}_N) \odot (N\mathbf{U}^{\ast}_N \mathbf{P}_{\mathrm{FC}} \mathbf{U}^{\ast}_N)\\
\label{eq:maskbeam_hdmd2}
&=\mathbf{X}\odot \mathbf{Z}.
\end{align}
The transformations that relate the matrices $\mathbf{H}$, $\mathbf{X}$, $\mathbf{G}$, and $\mathbf{S}$ are shown in Fig. \ref{fig:spec_mask_concept}. The matrix $\mathbf{S}$ is called the masked beamspace because it is an element-wise multiplication of the beamspace $\mathbf{X}$ and the spectral mask $\mathbf{Z}$. As the element-wise multiplication of a sparse matrix with any other matrix is a sparse matrix, $\mathbf{S}$ is sparse under the assumption that $\mathbf{X}$ is sparse. The vector $\mathbf{y}$ can be expressed using \eqref{eq:subsampled_conv}, \eqref{eq:modchannel} and \eqref{eq:maskbeam_defn} as
\begin{equation}
\label{eq:compactmaskedmeas}
\mathbf{y}=\mathcal{P}_{\Omega} (\mathbf{U}_N \mathbf{S} \mathbf{U}_N) + \mathbf{v}.
\end{equation}
The channel measurements in \eqref{eq:compactmaskedmeas} can be interpreted as the subsampled 2D-DFT of the masked beamspace $\mathbf{S}$. In a subsampling setting, i.e., $M<N^2$, the masked beamspace can be recovered from $\mathbf{y}$, using partial 2D-DFT CS techniques that exploit the sparsity of $\mathbf{S}$. %\footnote{Partial 2D-DFT CS is a special instance of the 2D-CCS technique that uses $\mathbf{e}_0\mathbf{e}^T_0$ as the base matrix instead of $\mathbf{P}$. For a Dirac-delta base matrix, $\mathbf{P}_{\mathrm{FC}}=\mathbf{e}_0\mathbf{e}^T_0$ and $\mathbf{G}=\mathbf{H}$.}
\begin{figure}[h]
\centering
\includegraphics[trim=0.75cm 0cm 0.75cm 0cm, width=0.25 \textwidth]{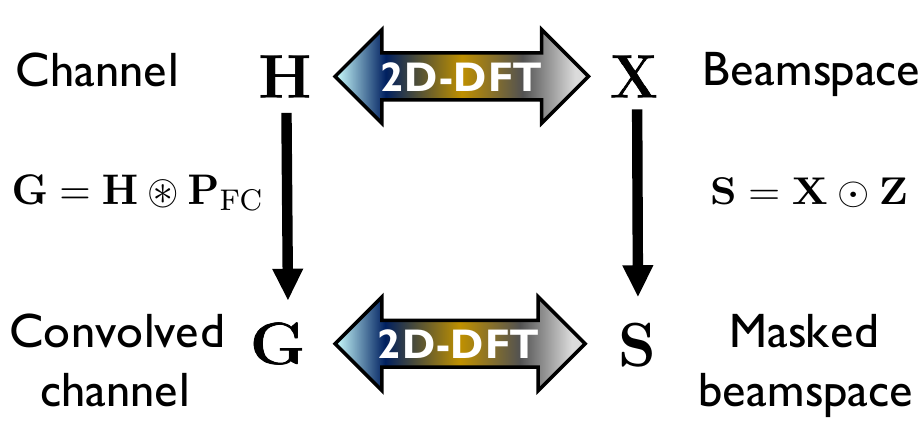}
\caption{\small A summary of the transformations relating the matrices $\mathbf{H}$, $\mathbf{X}$, $\mathbf{G}$, and $\mathbf{S}$. In this paper, we develop a base matrix $\mathbf{P}$ such that $\mathbf{H}$, or equivalently $\mathbf{X}$, can be efficiently estimated from 2D-CCS based channel measurements.}
\normalsize
  \label{fig:spec_mask_concept}
\end{figure}
\subsection{Conditions on the base matrix for efficient CS}\label{sec:conditions_base}
\par In this section, we derive guarantees on channel recovery for partial 2D-DFT CS over the masked beamspace. Using these guarantees, we identify the conditions on the base matrix for efficient 2D-CCS-based recovery of the beamspace channel $\mathbf{X}$. 
\par  The CS matrix that results from acquiring $M=\mathcal{O}(\mathrm{log}N)$ 2D-DFT samples of the sparse matrix $\mathbf{S}$ in \eqref{eq:compactmaskedmeas}, is known to satisfy the restricted isometry property with high probability \cite{kramer}. The masked beamspace $\mathbf{S}$ can be estimated from the channel measurements in \eqref{eq:compactmaskedmeas} using an $\ell_1$ optimization program \cite{csintro}
\begin{equation}
\hat{\mathbf{S}}=\!\! \begin{array}{c}
\mathrm{arg\,min\,\,}\Vert \mathbf{W}\Vert _{1} ,\,
\mathrm{s.t\,}\Vert \mathbf{y}-\mathcal{P}_{\Omega}(\mathbf{U}_{N}\mathbf{W}\mathbf{U}_{N})\Vert _{2}\leq\sqrt{M}\sigma
\end{array}.
\label{eq:l1opt}
\end{equation}
The optimization program in \eqref{eq:l1opt} encourages sparse masked beamspace solutions that are consistent with the received channel measurements \cite{csintro}. It is important to note that successful recovery of $\mathbf{S}$ does not guarantee the reconstruction of the beamspace channel, i.e., $\mathbf{X}$. The recovery of the beamspace depends on the spectral mask $\mathbf{Z}$. For example, if $\mathbf{Z}(k,\ell)=0$ for some $k$ and $\ell$, the masked beamspace component $\mathbf{S}(k,\ell)=0$. In such case, $\mathbf{X}(k,\ell)$ cannot be recovered from the spectral mask equation, i.e., $\mathbf{S}=\mathbf{Z} \odot \mathbf{X}$. To avoid such blanking effects in the masked beamspace, well conditioned spectral masks must be designed to estimate $\mathbf{X}$ from $\mathbf{S}$. 
\par We derive guarantees for compressive beamspace reconstruction using masked beamspace recovery with \eqref{eq:l1opt}. Let $\hat{\mathbf{X}}$ be a solution to the beamspace channel. As $\mathbf{S}=\mathbf{Z}\odot \mathbf{X}$, the estimate $\hat{\mathbf{X}}$ must satisfy $\hat{\mathbf{S}}=\mathbf{Z}\odot \hat{\mathbf{X}}$. For the spectral mask $\mathbf{Z}$, we define $\Zmax=\underset{k,\ell}{\mathrm{max}}\left|\mathbf{Z}(k,\ell)\right|$ and $\Zmin=\underset{k,\ell}{\mathrm{min}}\left|\mathbf{Z}(k,\ell)\right|$. We use $\left( \mathbf{A}\right)_k$ to denote the $k$ sparse representation of $\mathbf{A}$. The matrix $\left( \mathbf{A}\right)_k$ is obtained from $\mathbf{A}$ by retaining the $k$ largest entries in magnitude and setting the rest to $0$.
\begin{theorem} 
For a fixed constant $\gamma \in \left(0,1\right)$, a solution $\hat{\mathbf{X}}$ such that $\hat{\mathbf{S}}=\hat{\mathbf{X}} \odot \mathbf{Z}$ satisfies 
\begin{equation}
\bigl\Vert \mathbf{X}-\hat{\mathbf{X}}\bigl\Vert _{F}\leq C_{1}\frac{\Zmax\left\Vert \mathbf{X}-\left(\mathbf{X}\right)_{k}\right\Vert _{1}}{\sqrt{k}\Zmin}+C_{2}\frac{N\sigma}{\Zmin},
\label{eq:CSbound}
\end{equation}
with a probability of at least $1-\gamma$ if $M\geq Ck\,\mathrm{max}\left\{ 2\mathrm{log}^{3}(2k)\,\mathrm{log}(N),\,\mathrm{log}(\gamma^{-1})\right\}$. The constants $C, C_1$ and $C_2$ are independent of all the other parameters.
\label{theoremcondn}
\end{theorem} 
\begin{proof}
See Section \ref{sec:proof_condn}. 
\end{proof}
For a given $\mathbf{X}$, $M$, and $\sigma$, the result in \eqref{eq:CSbound} indicates that upper bound on the channel reconstruction error, i.e., $\bigl\Vert \mathbf{X}-\hat{\mathbf{X}}\bigl\Vert _{F}$, is lower when $\Zmin$ is bounded away from $0$. The smallest entry in $|\mathbf{Z}|$ depends on the base matrix $\mathbf{P}$. Note that $\mathbf{Z}$ is the inverse 2D-DFT of $\mathbf{P}_{\mathrm{FC}}$, the flipped and conjugated version of $\mathbf{P}$. As $\Zmin \leq \Zmax$, a base matrix that achieves the smallest upper bound in \eqref{eq:CSbound} is one that has $\Zmin=\Zmax$. %Therefore, a base matrix that results in a constant amplitude spectral mask results in a tight beamspace channel reconstruction guarantee according to \eqref{eq:CSbound}.
\par Now, we show that ideal base matrices in $\mathbb{Q}_q^{N \times N}$ must have a unimodular spectral mask, i.e., $\Zmax=1$ and $\Zmin=1$, for efficient 2D-CCS in phased arrays. It can be observed from \eqref{eq:specmask_defn} that the norm of the spectral mask is $\Vert\mathbf{Z}\Vert_{F}=N$ for any $\mathbf{P} \in \mathbb{Q}_q^{N \times N}$. The condition $\Zmin=\Zmax$ is achieved under the norm constraint only when all the entries of $|\mathbf{Z}|$ are equal to $1$. Designing a base matrix $\mathbf{P} \in \mathbb{Q}^{N \times  N}_q$ such that the spectral mask in \eqref{eq:specmask_defn} is unimodular, however, is a difficult problem. The main challenge in the design of $\mathbf{P}$ is due to the phase shift constraint, i.e., $\mathbf{P} \in \mathbb{Q}_q^{N \times N}$. The canonical basis element $\mathbf{e}_0\mathbf{e}^T_0$ is a good example that has a unimodular spectral mask, but lies outside  $ \mathbb{Q}_q^{N \times N}$. A brute force approach to find a matrix in $\mathbb{Q}_q^{N \times N}$ with a unimodular spectral mask is not practical as $\mathbb{Q}_q^{N \times N}$ contains a large number of matrices, i.e., $2^{qN^2}$. Prior work has considered subsampled convolution using random sequences \cite{kramer}; the 2D extension of such a technique is 2D-CCS using a $\mathbf{P}$ that is chosen at random from $\mathbb{Q}^{N \times N}_q$. A random choice for $\mathbf{P}$, however, may not result in a unimodular spectral mask. In Sec.~\ref{sec:why_perf_arrays}, we show that ideal base matrices exist for several combinations of $q$ and $N$.
\subsection{Perfect arrays as ideal base matrices} \label{sec:why_perf_arrays}
In this section, we establish the equivalence between unimodularity of the spectral mask and perfect periodic spatial autocorrelation of the base matrix. Using this equivalence, we show that perfect arrays \cite{PBA,PQA}, a class of matrices that have perfect periodic spatial autocorrelation, satisfy the properties of an ideal base matrix for 2D-CCS-based channel recovery. 
\par The duality between perfect periodic spatial autocorrelation and unimodular 2D-DFT properties is explained in Theorem \ref{theorem_array_duality}.% using the Dirac delta notation. 
\begin{theorem} 
A matrix $\mathbf{P} \in \mathbb{Q}^{N \times N}_q$ has a unimodular spectral mask, i.e., $\Zmax=1$ and $\Zmin=1$, if and only if $\mathbf{P}$ has perfect periodic spatial autocorrelation, i.e.,  
\begin{equation}
\label{eq:array_duality}
\langle \mathbf{P}, \mathbf{J}_x \mathbf{P} \mathbf{J}^T_y \rangle =0\;\; \forall \, (x,y)\in \mathcal{I}_N \times \mathcal{I}_N \setminus (0,0).
\end{equation}
\label{theorem_array_duality}
\end{theorem} 
\begin{proof}
%\vspace{-5mm}
See Section \ref{sec:proof_theorem_array_duality}.
\end{proof}
The problem of finding a perfect array in $\mathbb{Q}^{N \times N}_q$ over small alphabets, i.e., a small $q$, has been well investigated in \cite{PBA} and \cite{PQA}. Although there are several hardware constrained 2D-CS applications, perfect arrays over finite alphabets have not been used in the context of CS, to the best of our knowledge. The construction of perfect arrays over large alphabets, i.e., a large $q$, can be trivial. For example, a matrix that is an outer product of two Zadoff-Chu sequences satisfies the conditions in Theorem \ref{theorem_array_duality}, and is a perfect array. The ZC-based matrix, however, may not be realizable in low resolution phased arrays \cite{swiftlink,struct_rand_phased_array}. An important step towards efficient 2D-CCS in $q$-bit phased arrays, is to find matrices that are perfect arrays, i.e., matrices that satisfy the perfect periodic autocorrelation property in $\mathbb{Q}^{N\times N}_{q}$.
\par Perfect arrays over $\mathbb{Q}^{N_1 \times N_2}_q$ were constructed for binary and quaternary alphabets, i.e., for $q=\mathrm{log}_2{2}$ and $q=\mathrm{log}_2{4}$ in \cite{PBA} and \cite{PQA}. In this paper, we consider square arrays, i.e., arrays of size $N \times N$ for simplicity. We also consider the extreme case of perfect binary arrays, i.e., $q=1$, as such arrays can be implemented in phase shifters of any resolution. An example of a perfect binary array for $N=2$ is
\begin{equation}
\label{eq:PBAexample}
\mathbf{P}=\frac{1}{2} \left[\begin{matrix} 
1 & 1 \\
1 & -1 
\end{matrix}\right].
\end{equation}
The matrix in \eqref{eq:PBAexample} satisfies the perfect periodic autocorrelation property in \eqref{eq:array_duality}. The binary nature of $\mathbf{P}$ in \eqref{eq:PBAexample} allows its application to $2 \times 2$ phased arrays with a resolution of one-bit. As typical mmWave systems have large antenna arrays, it is useful to construct perfect arrays of large dimensions for efficient 2D-CCS. An interesting result from \cite{PBA} is that perfect binary arrays in $\mathbb{Q}^{N \times N}_1$ exist when $N=2^k$ or $N=3 \cdot 2^k$, where $k$ is any natural number. A recursive method to generate perfect binary arrays was provided in \cite{PBA} for square arrays and other rectangular configurations. An implementation of the construction in \cite{PBA} is available on our GitHub page \cite{PBA_implementation}. For $q=2$, several perfect arrays, for which perfect binary arrays of the same dimension do not exist, were proposed in \cite{PQA}. The arrays in \cite{PQA} can be used for efficient 2D-CCS in $2$-bit phased arrays. FALP uses perfect binary arrays in 2D-CCS, and allows the CS algorithm in \eqref{eq:l1opt} to achieve the smallest upper bound on the channel reconstruction error in \eqref{eq:CSbound}. 
\subsection{Perfect array-based compressive beam alignment in FALP} \label{sec:FALP_mainsubsection}
\par We show how channel matrices are estimated with FALP, using Fig. \ref{fig:summary_FALP}. FALP uses a partial 2D-DFT CS algorithm to estimate the masked beamspace matrix $\mathbf{S}$, from 2D-CCS-based channel measurements that are acquired with a perfect array. It can be observed from Theorem \ref{theorem_array_duality} that the spectral mask $|\mathbf{Z}|$ is unimodular for any perfect array, i.e., $|\mathbf{Z}(k, \ell)|=1\, \forall k, \ell$. In such a case, the transformation between the beamspace $\mathbf{X}$ and the masked beamspace $\mathbf{S}$ in \eqref{eq:maskbeam_hdmd2}, can be inverted using $\mathbf{X}(k,\ell)= \mathbf{S}(k,\ell) {\mathbf{Z}^{\text{c}}}(k,\ell) \, \forall k, \ell$. For a masked beamspace $\hat{\mathbf{S}}$ obtained from partial 2D-DFT CS, the beamspace matrix can be estimated as
\begin{equation}
\hat{\mathbf{X}}=\hat{\mathbf{S}} \odot {\mathbf{Z}^{\text{c}}}.
\end{equation}
The channel estimate $\hat{\mathbf{H}}=\mathbf{U}_N \hat{\mathbf{X}} \mathbf{U}_N$ is then used for beam alignment.  
\begin{figure}[h]
\centering
\includegraphics[trim=0.6cm 0cm 0.6cm 0cm, width=0.48 \textwidth]{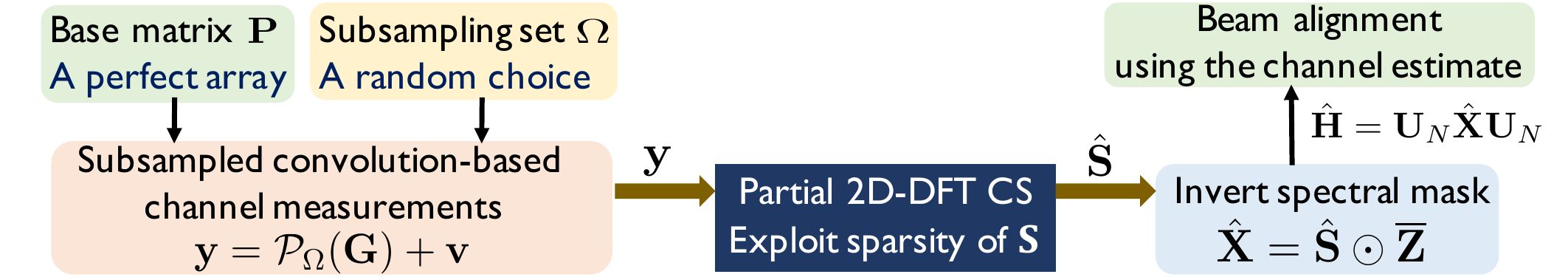}
\caption{\small Compressive channel acquisition and recovery using FALP. The channel measurements in FALP are acquired by applying 2D-circulant shifts of a perfect array according to a random subsampling set $\Omega$.}
\normalsize
  \label{fig:summary_FALP}
\end{figure}
\par Now, we explain a two step procedure for beam alignment with the channel estimate $\hat{\mathbf{H}}$. The first step relaxes the $q$-bit constraint to find an $\mathbf{F} \in \mathbb{Q}^{N \times N}_{\infty}$ that maximizes $|\langle \hat{\mathbf{H}},\mathbf{F}\rangle|$. Note that the phased array implementation requires $|\mathbf{F}_{k,\ell}|=1/N$ for every $k$ and $\ell$. By the dual norm inequality \cite{boyd2004convex}, we have $|\langle\hat{\mathbf{H}},\mathbf{F}\rangle|\leq\mathrm{max}(|\mathbf{F}|)\Vert\hat{\mathbf{H}}\Vert_{1}$. Therefore, $|\langle\hat{\mathbf{H}},\mathbf{F}\rangle|\leq \Vert\hat{\mathbf{H}}\Vert_{1}/N$. The upper bound in the dual norm inequality is achieved by an $\mathbf{F}^{\mathrm{opt}}(\beta)$ such that $|\mathbf{F}^{\mathrm{opt}}_{k,\ell}(\beta)|=1/N$ and $\mathrm{phase}(\mathbf{F}^{\mathrm{opt}}_{k,\ell}(\beta))=\beta + \mathrm{phase}(\hat{\mathbf{H}}_{k,\ell})$ for any $\beta \in (-\pi,\pi]$. The scalar $\beta$ corresponds to the global phase in $\mathbf{F}^{\mathrm{opt}}(\beta)$. As the angles in $\mathbf{F}^{\mathrm{opt}}(\beta)$ may not be integer multiples of $2\pi/2^q$, $\mathbf{F}^{\mathrm{opt}}(\beta)$ may not be directly realized in $q$-bit phased arrays. In such a case, a $q$-bit phase quantized version of $\mathbf{F}^{\mathrm{opt}}(\beta)$ can be used in the phased array, for an appropriate choice of global phase $\beta$.
 \par The second step of our beam alignment procedure finds the best $\beta$ that minimizes phase errors due to $q$-bit phase quantization of $\mathbf{F}^{\mathrm{opt}}(\beta)$. This step is important in low resolution phased arrays \cite{bias_phase,wang2018hybrid}. Let $\mathcal{Q}_q(\mathbf{F}^{\mathrm{opt}}(\beta))$ denote the $q$-bit phase quantized version of $\mathbf{F}^{\mathrm{opt}}(\beta)$. Note that $\mathcal{Q}_q(\cdot)$ performs element-wise phase quantization. The global phase term $\beta_{\mathrm{est}}$ that minimizes the phase quantization error can be expressed as
\begin{equation}
\label{eq:beta_est}
\beta_{\mathrm{est}}=\underset{\beta\in{(0,2\pi/2^{q})}}{\mathrm{argmin}}\,\Vert\mathcal{Q}_{q}(\mathbf{F}^{\mathrm{opt}}(\beta))-\mathbf{F}^{\mathrm{opt}}(\beta)\Vert_{F}.
\end{equation}
To solve for the scalar $\beta_{\mathrm{est}}$ in \eqref{eq:beta_est}, we define a phase set $\mathcal{B}$ that contains $K_{\mathcal{B}}$ uniformly spaced values in $(0,2\pi/2^{q})$. The optimization in \eqref{eq:beta_est} is performed using line search over the elements in $\mathcal{B}$ for a sufficiently large $K_{\mathcal{B}}$. The phase shift matrix used at the TX with CS-based beamforming is then $\mathbf{F}_{\mathrm{CS}}=\mathcal{Q}_{q}(\mathbf{F}^{\mathrm{opt}}(\beta_{\mathrm{est}}))$.
\par The partial 2D-DFT CS algorithm in FALP requires a lower computational complexity when compared to other standard CS techniques. Let $\mathbf{A}_{\mathrm{CS}} \in \mathbb{C}^{M\times N^2}$ be the CS matrix corresponding to the partial 2D-DFT CS problem in \eqref{eq:compactmaskedmeas}. CS algorithms that solve \eqref{eq:compactmaskedmeas} typically perform iterative optimization over an $N^2$ dimensional variable. For example, each iteration in the orthogonal matching pursuit (OMP) algorithm \cite{OMP_ref} requires computing matrix-vector products of the form $\mathbf{A}_{\mathrm{CS}} \mathbf{w}$ and  $\mathbf{A}^{\ast}_{\mathrm{CS}} \mathbf{d}$. As $\mathbf{A}_{\mathrm{CS}}$ is a partial 2D-DFT CS matrix for the model in \eqref{eq:compactmaskedmeas}, the matrix-vector products in CS can be implemented using the 2D-FFT \cite{partial2DDFTCS}. In the subsampling regime where $M=\alpha N^2$ for some constant $\alpha <1$, the 2D-FFT-based implementation has a complexity of $\mathcal{O}(N^2 \, \mathrm{log}N)$ while the complexity of standard matrix-vector product is  $\mathcal{O}(N^4)$ \cite{partial2DDFTCS}. 
\section{An MRI-inspired approach for ultra-low complexity beam alignment }\label{sec:MRI_BA}
In this section, we use insights from zero filling reconstruction in MRI \cite{bernstein2001effect}, to develop a different sub-Nyquist beam alignment technique based on FALP. The proposed technique does not require any iterative optimization, unlike standard CS or partial 2D-DFT CS. Specifically, the zero filling-based approach uses the perfect array-based training, and estimates a reasonable beamformer with just a single 2D-FFT computation. We prove that our method can achieve a beam alignment performance that is comparable to exhaustive scan with the 2D-DFT dictionary.
\subsection{Connection between CS in MRI and CS using FALP}
\par The measurements in MRI are defined by a trajectory that acquires samples from the Fourier transform of an MR image, also known as the k-space \cite{bernstein2001effect}. Prior work on CS-MRI has shown that MR images can be reconstructed using subsampling k-space trajectories that acquire fewer samples from the k-space \cite{sparseMRI}. For example, CS can reconstruct sparse angiogram images from fewer samples of their 2D-DFT \cite{csintro}. In this section, we explain the equivalent of k-space trajectory in FALP.
\par The channel measurements in FALP, i.e., $\mathcal{P}_{\Omega}(\mathbf{G})$, are samples from the 2D-DFT of the sparse masked beamspace $\mathbf{S}$. The subsampling pattern over $\mathbf{G}$ is determined by the set $\Omega$, as shown in Fig. \ref{fig:2dccs_illus}. The masked beamspace $\mathbf{S}$  is analogous to sparse angiogram image in MRI. The k-space, which represents the Fourier transform of the MR image, is equivalent to the matrix $\mathbf{G}$ as $\mathbf{G}=\mathbf{U}_N \mathbf{S}\mathbf{U}_N$. The analogue of a k-space trajectory in FALP is a 2D-curve in $\mathcal{I}_N \times \mathcal{I}_N$ that sequentially traverses through the coordinates in $\Omega$. An example of a trajectory for $M=9$ and $N=5$ is shown in Fig. \ref{fig:falptraj}. The trajectory in Fig. \ref{fig:falptraj} sequentially acquires the channel measurements $\{\mathbf{G}_{01}, \mathbf{G}_{22},\mathbf{G}_{04}, \cdots, \mathbf{G}_{32} \}$ for the subsampling set $\Omega=\{(0,1),(2,2),(0,4), \cdots , (3,2)\}$.
\begin{figure}[htbp]
\centering
\subfloat[Example of a trajectory]{\includegraphics[width=4.25cm, height=4.25cm]{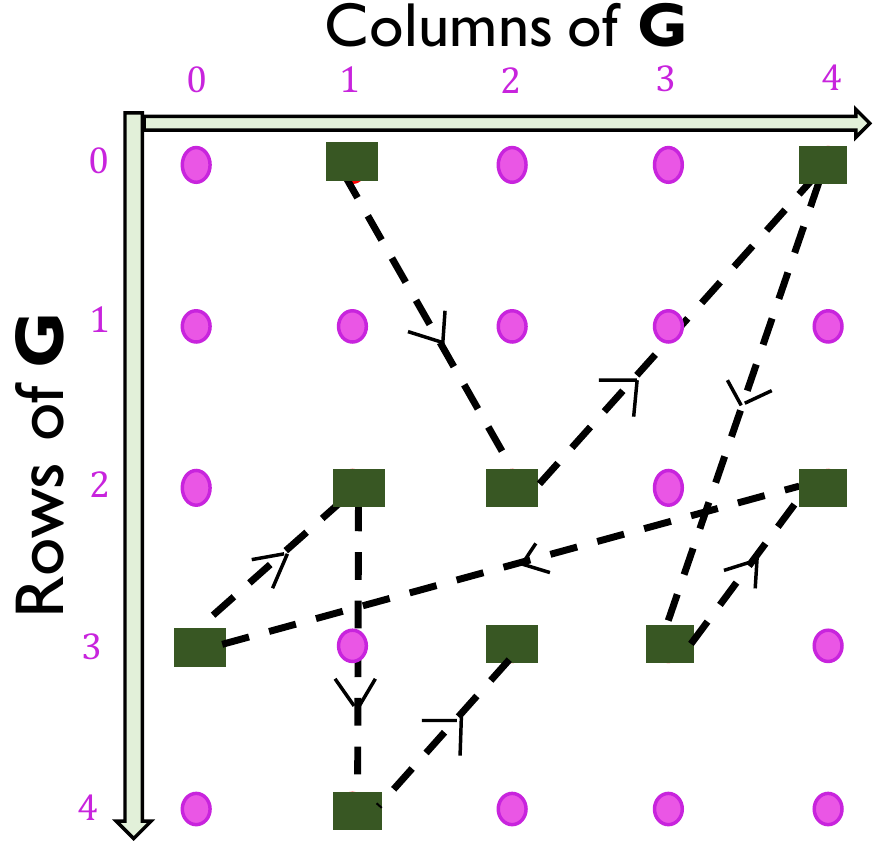}\label{fig:falptraj}}
\:\:\:
\subfloat[A realization of $|\mathbf{K}_{\mathrm{bl}}|$]{\includegraphics[trim=0cm 0cm 0cm 0cm,clip=true,width=4.25cm, height=4.25cm]{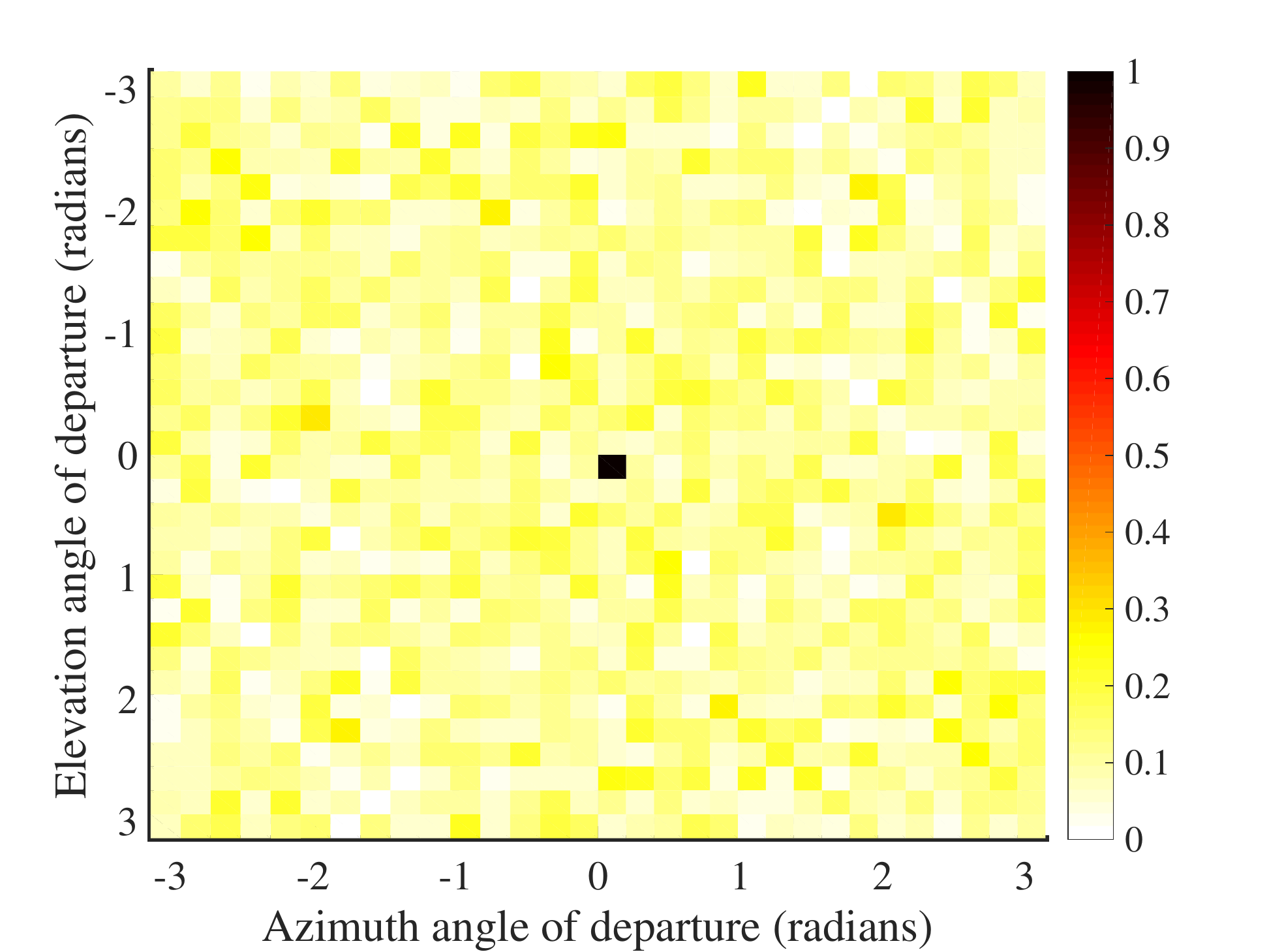}\label{fig:psf_rho_1_16}}
\caption{ \small The CS matrix in FALP is determined by a trajectory in the $\mathbf{G}$-space. The kernel matrix $\mathbf{K}_{\mathrm{bl}}$ is not a perfect point spread function due to subsampling. The matrix $\mathbf{K}_{\mathrm{bl}}$ shown in this example is for $N=32$ and $\rho = 1/16$.\normalsize}
\end{figure}
We would like to mention that k-space trajectories in MRI are typically constrained to be continuous, i.e., random k-space trajectories may not be realized. Random trajectories over the matrix $\mathbf{G}$, however, can be realized in FALP as any circulant shift of a base matrix can be applied to the phased array. In this paper, we use ideas from CS-MRI to investigate how beam alignment can be performed without any iterative optimization.
\subsection{Zero filling-based reconstruction: From MRI to beam alignment}
A traditional approach for MR imaging is to use a trajectory that fully samples the k-space. The MR image is then recovered by applying a 2D-Fourier transform over the acquired samples. Zero filling-based technique is an approach to estimate MR images using trajectories that subsample the k-space. The idea in zero filling-based reconstruction is to fill the unsampled entries in the k-space with zeros, and invert the Fourier transform to estimate the MR image. In this section, we show that zero filling-based recovery can also be used to recover a reasonable one-sparse approximation of the beamspace. To provide a better illustration and for a tractable analysis, we ignore the measurement noise in the system, i.e., $\sigma = 0$. The simulation results in Section \ref{sec:simulations} include the impact of measurement noise.
\par Now, we explain zero filling-based beamspace reconstruction using FALP. We define the subsampling ratio in FALP as $\rho=M/N^2$. The matrix $\mathbf{G}_{\Omega} \in \mathbb{C}^{N\times N}$ is defined to contain the entries of $\mathbf{G}$ at the locations in $\Omega$, and zeros in the other locations. Specifically, the entries of $\mathbf{G}_{\Omega}$ are given by
\begin{align}
\label{eq:ddualsetnz}
\mathbf{G}_{\Omega}(r[m],c[m])&= \mathbf{G}(r[m],c[m])\,\,\, \mathrm{for}\,\,1\leq m\leq M, \,\, \mathrm{and}\\
\label{eq:ddualsetz}
\mathbf{G}_{\Omega}(k,\ell)&=0 \,\,\,\,\,\forall(k,\ell)\notin\Omega.
\end{align}
The matrix $\mathbf{G}_{\Omega}$ can be constructed by populating the $M$ channel measurements in FALP at the locations in $\Omega$. In a subsampling setting, i.e., $M<N^2$, the construction of $\mathbf{G}_{\Omega}$ is equivalent to zero filling in MRI. The equivalence follows from the observation that $\mathbf{G}_{\Omega}$ is $0$ at the unsampled $\mathbf{G}$-space locations. In a full sampling setting, the beamspace $\mathbf{X}$ can be estimated from the transformations $\mathbf{S}=\mathbf{U}^{\ast}_N \mathbf{G} \mathbf{U}^{\ast}_N$ and $\mathbf{X}=\mathbf{S} \odot {\mathbf{Z}^{\text{c}}}$. The zero filling-based reconstruction procedure applies the same transformations over $\mathbf{G}_{\Omega}$, when $M<N^2$. We define the zero filling-based estimates corresponding to $\mathbf{S}$ and $\mathbf{X}$ as
\begin{align}
\label{eq:sbl_defn}
\mathbf{S}_{{\mathrm{bl}}}&=\mathbf{U}^{\ast}_N \mathbf{G}_{\Omega} \mathbf{U}^{\ast}_N \,\, \mathrm{and}\\
\label{eq:xbl_defn}
\mathbf{X}_{{\mathrm{bl}}}&=\mathbf{S}_{{\mathrm{bl}}} \odot {\mathbf{Z}^{\text{c}}}.
\end{align}
Note that $\mathbf{X}_{{\mathrm{bl}}}=\mathbf{X}$, when $M=N^2$. A natural question that arises is how does subsampling impact the zero filling-based estimate $\mathbf{X}_{{\mathrm{bl}}}$ when compared to $\mathbf{X}$.
\par Now, we show that subsampling the $\mathbf{G}$-space results in a blurred $\mathbf{X}_{{\mathrm{bl}}}$. For ease of notation, we investigate the impact of blur on $\mathbf{S}_{{\mathrm{bl}}}$; our study is justified by the fact that $\mathbf{X}_{{\mathrm{bl}}}$ in \eqref{eq:xbl_defn} is just a phase modulated version of $\mathbf{S}_{{\mathrm{bl}}}$ for a unimodular $\mathbf{Z}$. To characterize the impact of subsampling on $\mathbf{S}_{{\mathrm{bl}}}$, we define a binary matrix $\mathbf{N}_{\Omega} \in \mathbb{C}^{N \times N}$ such that 
\begin{equation}
\label{eq:binary_mat}
\mathbf{N}_{\Omega}(r,c)=\begin{cases}
\begin{array}{c}
1,\,\,\,\,\mathrm{if}\,\,(r,c)\in\Omega\\
0,\,\,\,\,\mathrm{if\,\,(r,c)\notin\Omega}
\end{array}\end{cases}.
\end{equation}
We define a kernel matrix $\mathbf{K}_{\mathrm{bl}}$ as the scaled inverse 2D-DFT of $\mathbf{N}_{\Omega}$, i.e., 
\begin{equation}
\label{eq:kbldefn}
\mathbf{K}_{\mathrm{bl}}=\frac{N}{M} \mathbf{U}^{\ast}_N \mathbf{N}_{\Omega} \mathbf{U}^{\ast}_N.
\end{equation}
It can be observed from \eqref{eq:binary_mat} that $\mathbf{G}_{\Omega}=\mathbf{G} \odot \mathbf{N}_{\Omega}$. As element-wise multiplication of two matrices results in 2D-circular convolution of their inverse 2D-DFTs \cite{imageprocess}, \eqref{eq:sbl_defn} can be simplified to
\begin{align}
\mathbf{S}_{{\mathrm{bl}}}&=(\mathbf{U}_{N}^{\ast}\mathbf{G}\mathbf{U}_{N}^{\ast})\circledast(\mathbf{U}_{N}^{\ast}\mathbf{N}_{\Omega}\mathbf{U}_{N}^{\ast})/N\\
\label{eq:convdistort}
&=\frac{M}{N^2}\mathbf{S} \circledast \mathbf{K}_{\mathrm{bl}}.
\end{align}
It can be observed from \eqref{eq:convdistort} that the matrices $\mathbf{S}_{{\mathrm{bl}}}$ and $\mathbf{S}$ differ by a convolutional distortion due to $\mathbf{K}_{\mathrm{bl}}$. The matrix $\mathbf{K}_{\mathrm{bl}}$ is similar to the point spread function (PSF) in MRI. For the special case of $M=N^2$, it can be observed that $\mathbf{K}_{\mathrm{bl}}=\mathbf{e}_0 \mathbf{e}^T_0$ and $\mathbf{S}_{{\mathrm{bl}}}= \mathbf{S}$. In the subsampling regime, however, the PSF $\mathbf{K}_{\mathrm{bl}}$ is not a perfect dirac matrix. Therefore, $\mathbf{K}_{\mathrm{bl}}$ induces distortion in $\mathbf{S}_{{\mathrm{bl}}}$ for $M<N^2$. For the masked beamspace in Fig. \ref{fig:orig_S}, an example of the distortion induced due to subsampling is shown in Fig. \ref{fig:distorted_S}. The amount of distortion in $\mathbf{S}_{{\mathrm{bl}}}$ is a function of the subsampling ratio, i.e., $\rho$. 
\begin{figure}[htbp]
%\vspace{-5mm}
\centering
\subfloat[A two-sparse $|\mathbf{S}|$]{\includegraphics[width=4.25cm, height=4.25cm]{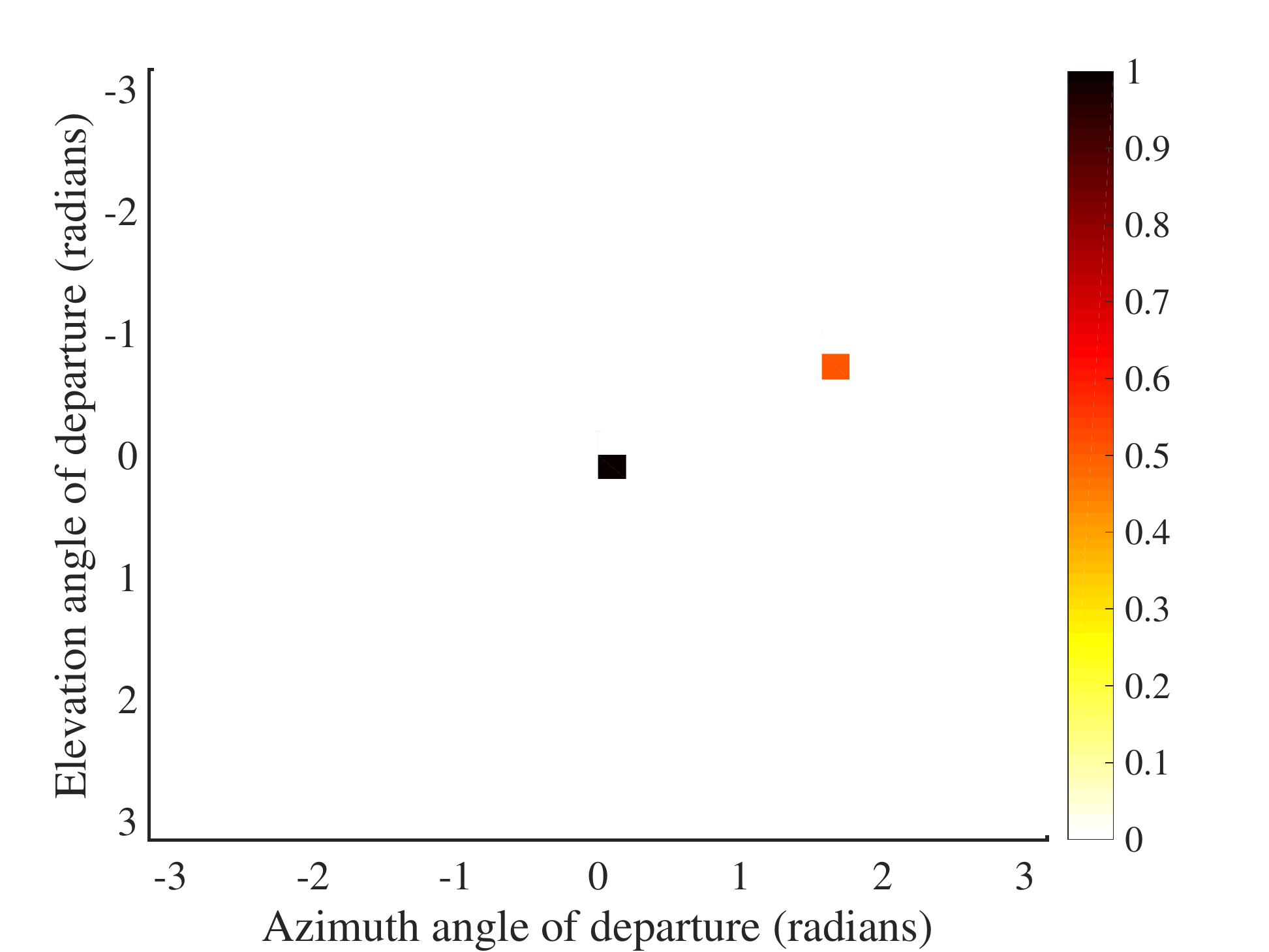}\label{fig:orig_S}}
\:\:\:
\subfloat[$|\mathbf{S}\circledast \mathbf{K}_{\mathrm{bl}}|$ for $\mathbf{K}_{\mathrm{bl}}$ in Fig. \ref{fig:psf_rho_1_16}]{\includegraphics[trim=0cm 0cm 0cm 0cm,clip=true,width=4.25cm, height=4.25cm]{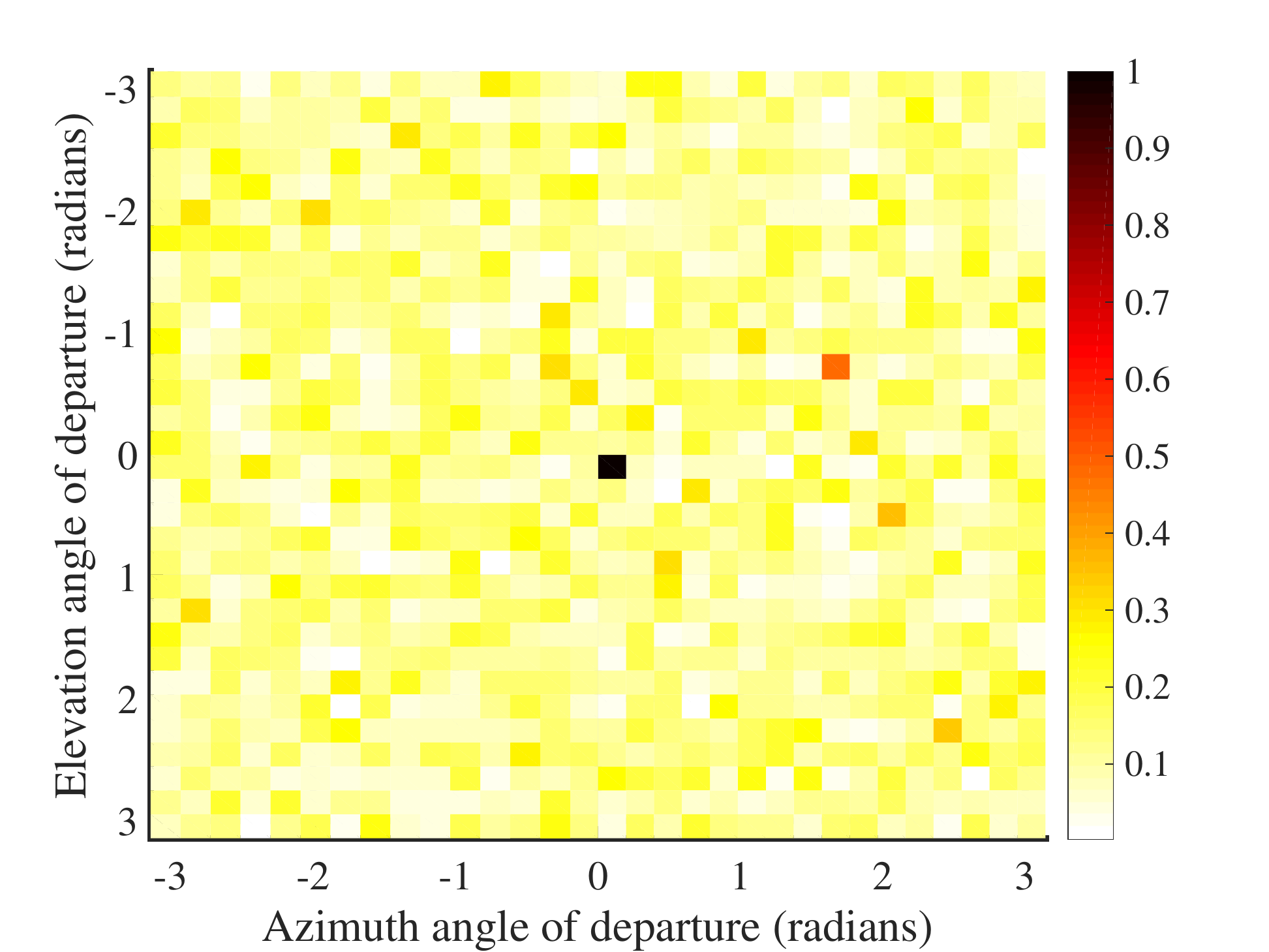}\label{fig:distorted_S}}
\caption{ \small The distortion in $\mathbf{S}\circledast \mathbf{K}_{\mathrm{bl}}$ depends on the subsampling set $\Omega$. From Fig. \ref{fig:orig_S} and Fig. \ref{fig:distorted_S}, it can be observed that coordinate corresponding to the best direction remains unchanged after distortion. Here, $\rho=1/16$ and $N=32$.\normalsize}
\end{figure}
\subsection{Beam alignment with the zero filling-based estimate} \label{sec:ZFB_procedure}
We define the zero filling-based beam alignment technique (ZFB) as one that chooses the beamformer based on the coordinate that maximizes $|\mathbf{S}_{{\mathrm{bl}}}|$, i.e., the zero filling-based masked beamspace estimate. If $|\mathbf{S}_{{\mathrm{bl}}}|$ achieves its maximum at $(r_{_{\mathrm{ZFB}}},c_{_{\mathrm{ZFB}}})$, the transmit beamformer in ZFB is defined as 
\begin{equation}
\mathbf{F}_{_{\mathrm{ZFB}}}=\mathcal{Q}_q \left(\mathbf{U}_N(:,r_{_{\mathrm{ZFB}}})\mathbf{U}_N(c_{_{\mathrm{ZFB}}},:) \right).
\end{equation}
It is important to note that exhaustive scan with the 2D-DFT dictionary selects the coordinate that maximizes $|\mathbf{X}|$, which is same as the one that maximizes $|\mathbf{S}|$, for a unimodular $\mathbf{Z}$. In an ideal setting, i.e., $M=N^2$, $\mathbf{S}_{{\mathrm{bl}}}=\mathbf{S}$ and ZFB results in the same beamformer as exhaustive scan with the 2D-DFT dictionary. In this section, we identify the subsampling regime for which ZFB results in the same beamformer as exhaustive scan.
\par Zero filling-based beam alignment is successful when the coordinate that maximizes $\mathbf{S}_{{\mathrm{bl}}}$ also maximizes $\mathbf{S}$. As $\mathbf{S}_{{\mathrm{bl}}}= \mathbf{S} \circledast \mathbf{K}_{{\mathrm{bl}}}$, it is important to characterize the entries of $\mathbf{K}_{{\mathrm{bl}}}$ to determine the success of ZFB. The matrix $\mathbf{K}_{{\mathrm{bl}}}$ in \eqref{eq:kbldefn} is a function of the subsampling set $\Omega$, that is chosen at random. It can be observed from \eqref{eq:binary_mat} that $\mathbf{N}_{\Omega}$ has $M$ ones and $N^2-M$ zeros. Therefore, $\mathbf{K}_{\mathrm{bl}}(0,0)$, the scaled DC-component of $\mathbf{N}_{\Omega}$, is $1$. The other entries of $\mathbf{K}_{\mathrm{bl}}$ explicitly depend on the elements in the sampled set $\Omega$ unlike $\mathbf{K}_{\mathrm{bl}}(0,0)$. As $\Omega$ is sampled at random, $\mathbf{K}_{\mathrm{bl}}(r,c)$ can be modelled as a random variable for any $(r,c)\neq (0,0)$ with a variance \cite{sparseMRI}
\begin{equation}
\label{eq:var_defn}
\xi^{2}=\frac{1-\rho}{\rho N^{2}}.
\end{equation} 
Note that the variance of the PSF at the $N^2-1$ locations other than $(0,0)$ is exactly the same when $\Omega$ is chosen uniformly at random. The magnitude of $|\mathbf{K}_{\mathrm{bl}}|$, is shown in Fig. \ref{fig:psf_rho_1_16} for a particular realization of $\Omega$. 
 %It can be observed from Fig. \ref{fig:orig_S} and \ref{fig:distorted_S} that $\mathbf{S} \circledast \mathbf{K}_{\mathrm{bl}}$ is a distorted version of $\mathbf{S}$. The distortion in $\mathbf{S} \circledast \mathbf{K}_{\mathrm{bl}}$ is lesser for a higher subsampling ratio $\rho$.
\par We explain the setting used to investigate beam alignment with the zero-filling based approach. For simplicity of analysis, we consider a $2$-path channel such that the beamspace angles of departure of each path are aligned with those defined by the 2D-DFT dictionary. Without loss of generality, we consider $\mathbf{S}(0,0)=1$ and $\mathbf{S}(r_o,c_o)=a$, such that $|a|<1$ and $(r_o,c_o)$ is some coordinate other than $(0,0)$. The remaining $N^2-2$ entries of $\mathbf{S}$ are equal to $0$ as seen in Fig. \ref{fig:orig_S}. For such a setting, beam alignment via ZFB is successful when $|(\mathbf{S} \circledast \mathbf{K}_{\mathrm{bl}})_{0,0}|$ is the largest entry in $|\mathbf{S} \circledast \mathbf{K}_{\mathrm{bl}}|$. The probability that ZFB  is successful can be expressed as 
\begin{equation}
\label{eq:pexact}
p=\mathrm{Pr}(\underset{(r,c)\neq(0,0)}{\cap}|(\mathbf{S}\circledast\mathbf{K}_{\mathrm{bl}})_{0,0}|>|(\mathbf{S}\circledast\mathbf{K}_{\mathrm{bl}})_{r,c}|).
\end{equation}  
The statistics of the PSF, i.e., $\mathbf{K}_{\mathrm{bl}}$, can be used to determine the entries in $\mathbf{S} \circledast \mathbf{K}_{\mathrm{bl}}$. Prior work in MRI \cite{sparseMRI} and partial 2D-DFT CS \cite{2D_CS_gaussian_noise} has modelled $\mathbf{K}_{\mathrm{bl}}(r,c)$ as $\mathcal{N}_{\mathrm{c}}(0, \xi^2)$ for any $(r,c)\neq (0,0)$. It is important to note that the random variables $\mathbf{K}_{\mathrm{bl}}(r_1,c_1)$ and $\mathbf{K}_{\mathrm{bl}}(r_2,c_2)$ can be coupled for $(r_1,c_1)\neq (0,0)$ and $(r_2,c_2)\neq (0,0)$. For instance, as $\mathbf{N}_{\Omega}$ is a real matrix, its inverse 2D-DFT must be conjugate symmetric \cite{imageprocess}, i.e., $\mathbf{K}_{\mathrm{bl}}(N-r,N-c)={\mathbf{K}^{\text{c}}}_{\mathrm{bl}}(r,c)$ for any $(r,c) \in \mathcal{I}_{N} \times \mathcal{I}_{N}$. In this paper, we account for such dependencies to derive a lower bound on the probability of success with ZFB.
\par Now, we describe the statistics of the entries in $\mathbf{S} \circledast \mathbf{K}_{\mathrm{bl}}$. For distinct coordinates $(r_o,c_o)$ and $(r_1,c_1)$, we model $\mathbf{K}_{\mathrm{bl}}(r_o,c_o)$, $\mathbf{K}_{\mathrm{bl}}(r_1,c_1)$ and $\mathbf{K}_{\mathrm{bl}}(r_1-r_o,c_1-c_o)$ as IID random variables $\mathsf{x}$, $\mathsf{b}$ and $\mathsf{w}$. To validate the independence assumption, we conducted empirical studies on the total variation distance between the joint distribution of $\mathsf{x}$, $\mathsf{b}$ and $\mathsf{w}$, and the product of their marginals. Our studies indicate that the variables can be considered ``independent'' except for the case when $(r_1,c_1)=(2r_o,2c_o)$ or $((r_o+c_o)/2, (r_o+c_o)/2)$ or $(-r_o,-c_o)$. We ignore these three scenarios to assume that $\mathsf{x}$, $\mathsf{b}$ and $\mathsf{w}$ are independent; such an assumption was also made in \cite{2D_CS_gaussian_noise}. Each of the random variables $\mathsf{x}$, $\mathsf{b}$ and $\mathsf{w}$ is distributed as $\mathcal{N}_{\mathrm{c}}(0,\xi^2)$. From the conjugate symmetry property, it follows that $\mathbf{K}_{\mathrm{bl}}(N-r_o,N-c_o)=\mathsf{x}^{\ast}$. The $(k,\ell)^{\mathrm{th}}$ entry in $\mathbf{S} \circledast \mathbf{K}_{\mathrm{bl}}$ can be expressed as 
\begin{equation}
\label{eq:conv_compact}
(\mathbf{S} \circledast \mathbf{K}_{\mathrm{bl}})_{k,\ell}=\sum^{N-1}_{r=0}\ \sum^{N-1}_{c=0}\mathbf{S}(r,c)\mathbf{K}_{\mathrm{bl}}\left((k-r)_N,(\ell-c)_N\right).
\end{equation}
%If a unit entry at $\mathbf{S}(0,0)$ contributes to an interference of $\mathsf{x}$ in $(\mathbf{S} \circledast \mathbf{K}_{\mathrm{bl}})_{r_o,c_o}$, it follows from \eqref{eq:conv_compact} that a unit entry at $\mathbf{S}(r_o,c_o)$ will result in an interference of $\mathsf{x}^{\ast}$ in $(\mathbf{S} \circledast \mathbf{K}_{\mathrm{bl}})_{0,0}$. 
For a $2$-sparse $\mathbf{S}$ defined in this section, it can be shown from \eqref{eq:conv_compact} that $(\mathbf{S} \circledast \mathbf{K}_{\mathrm{bl}})_{0,0}$, $(\mathbf{S} \circledast \mathbf{K}_{\mathrm{bl}})_{r_o,c_o}$ and $(\mathbf{S} \circledast \mathbf{K}_{\mathrm{bl}})_{r_1,c_1}$ are $1+a\mathsf{x}^{\ast}$, $a+\mathsf{x}$ and $\mathsf{b}+a\mathsf{w}$.
\par We now derive a lower bound on the probability of successful beam alignment for the $2-$ sparse channel. We define $\mathcal{Q}_1(\alpha, \beta)$ as the first order Marcum-Q function with parameters $\alpha$ and $\beta$ \cite{chi_sq}. A lower bound on the beam alignment probability in \eqref{eq:pexact} is derived in Theorem \ref{theorem_prob_bound}.
\begin{theorem} 
For a $2-$ sparse beamspace channel, the probability that zero filling-based beam alignment is successful can be lower bounded as  
\begin{align}
\nonumber
p&\geq1-\mathcal{Q}_{1}(0,N \sqrt{{2 \rho}/({1-\rho}) })\\
\label{eq:probab_result}
&-\frac{(1+|a|^{2})(N^{2}-2)}{1+2|a|^{2}}\mathrm{exp}\left(\frac{-N^2 \rho}{(1+2|a|^{2})(1-\rho)}\right).
\end{align}
% Include lemma on sublinear rate is sufficient ?
\label{theorem_prob_bound}
\end{theorem} 
\begin{proof}
See Section \ref{sec:proof_prob_bound}.
\end{proof}
\par We show that the phase transition region that follows from our bound in \eqref{eq:probab_result} matches well with that observed from simulations. We consider a setting with $N=32$, and a $2-$ sparse $\mathbf{S}$ with $\mathbf{S}(0,0)=1$ and $\mathbf{S}(r_o,c_o)=a$ for some $|a|<1$. In our simulations, $(r_o,c_o)$ was chosen uniformly at random, and $M$ entries of $\mathbf{G}$ were sampled at random by applying $M$ random circulant shifts of a $32\times 32$ perfect binary array at the TX. Beam alignment is declared successful if the zero filling-based estimate, i.e., $|\mathbf{S} \circledast \mathbf{K}_{\mathrm{bl}}|$, achieves its maximum at $(0,0)$. The phase transition regions corresponding to the bound in \eqref{eq:probab_result} and the observed beam alignment probability are shown in Fig. \ref{fig:ps_anal} and Fig. \ref{fig:ps_sim}. The plots indicate that ZFB performs well with sub-Nyquist channel measurements. The proposed zero filling-based technique estimates a good one-sparse approximation of $\mathbf{S}$ that is consistent with the channel measurements. CS algorithms can estimate better sparse approximations of $\mathbf{S}$, but require iterative optimization. ZFB can provide a reasonable beamformer with a single 2D-FFT, and is advantageous over CS-based beam alignment in terms of computational complexity.
\begin{figure}[htbp]
%\vspace{-5mm}
\centering
\subfloat[Phase transition region using \eqref{eq:probab_result}.]{\includegraphics[width=4.25cm, height=4.25cm]{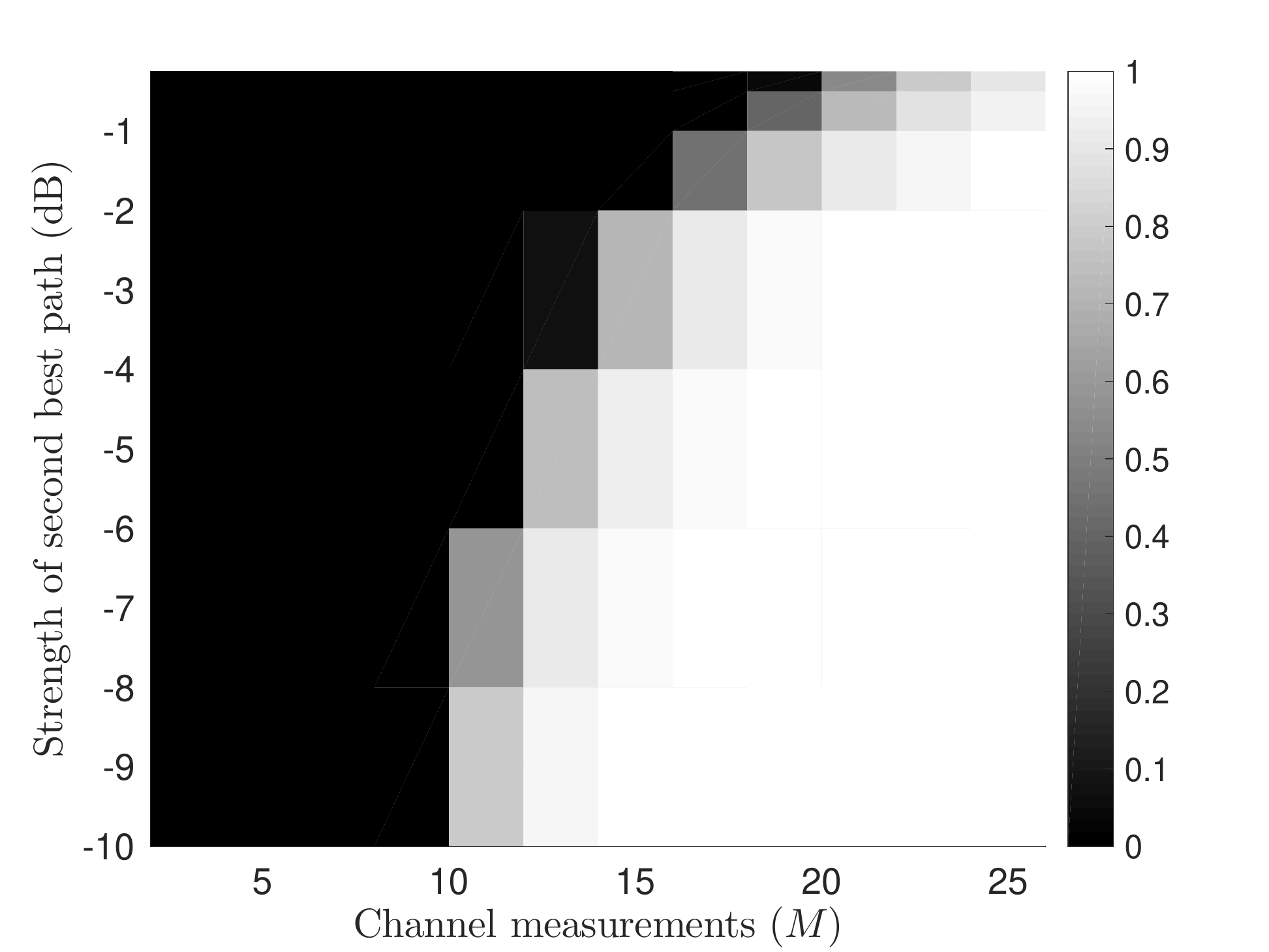}\label{fig:ps_anal}}
\:\:\:
\subfloat[Observed phase transition region]{\includegraphics[trim=0cm 0cm 0cm 0cm,clip=true,width=4.25cm, height=4.25cm]{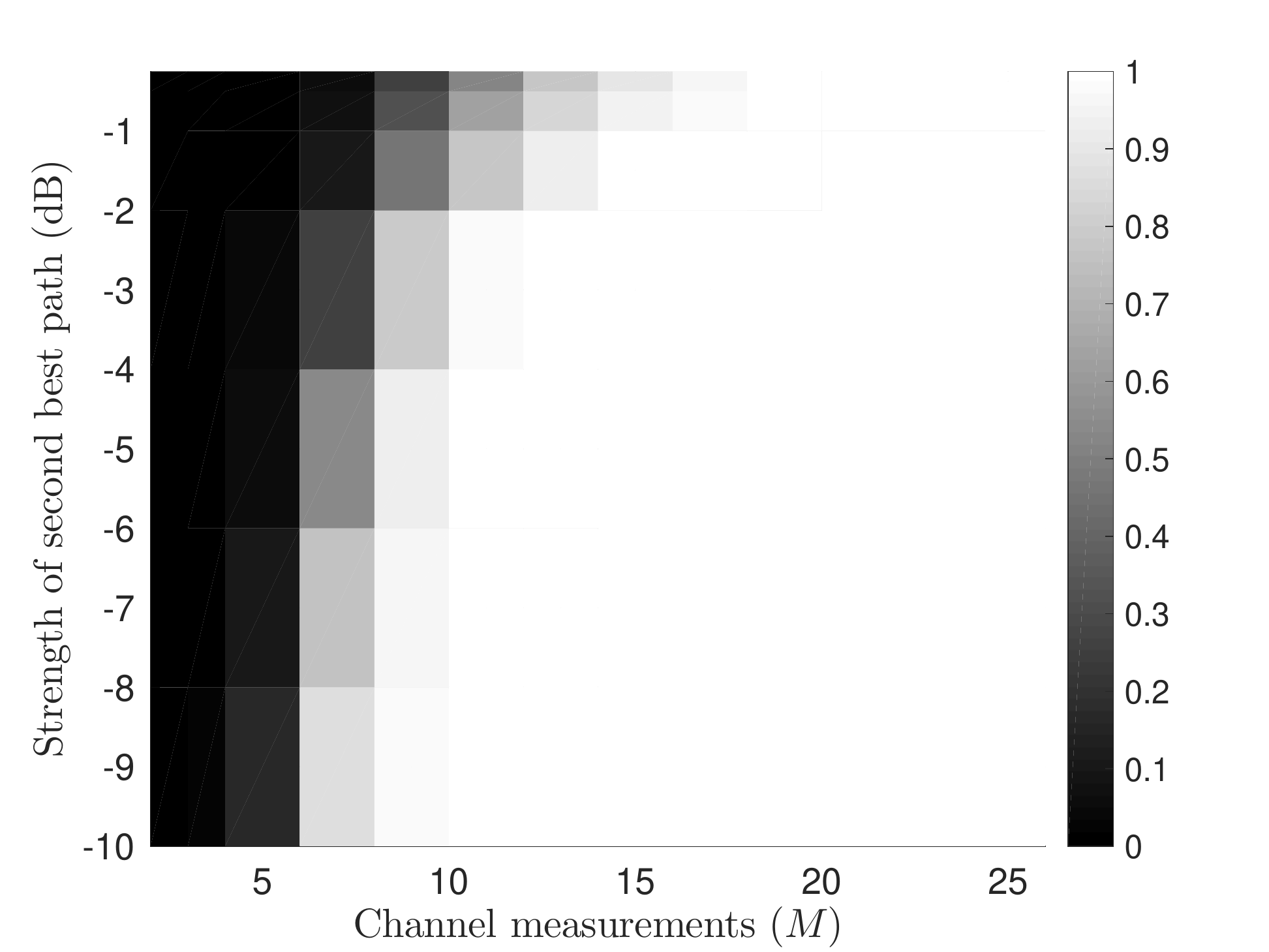}\label{fig:ps_sim}}
\caption{ \small For a $2$-sparse channel, the plots show the beam alignment probability as a function of the strength of second best path and the number of channel measurements. Here, $N=32$, $\rho=M/1024$ and $\sigma=0$. The strength of the second best path is  $20\,\mathrm{log}_{10}(|a|)$. In this example, ZFB succeeds with a high probability for $\rho \geq 0.03$. \normalsize}
\end{figure}
%\par From a signal recovery perspective, zero filling reconstruction results in noisy MR images for $\rho<1$. For beam alignment, however, the noisy beamspace matrix for $\rho<1$ can be good enough. For instance, the best direction for transmit beam alignment can still be determined from the noisy beamspace matrix in Fig. \ref{fig:distorted_S}. The connection between beam alignment using FALP and MRI opens an interesting research direction in compressive beamforming.  Fo k-space trajectories in MRI can be used to design CS matrices in mmWave systems that achieve robustness to non-idealities like CFO, phase noise, and frame synchronization errors. %For instance, it can be shown that CFO in mmWave systems \cite{swiftlink} is analogous to off-resonance in MRI \cite{smith2010mri}. As a result, k-space trajectories and algorithms in MRI that are robust to off-resonance can be applied for CFO robust beamforming using FALP. 
\section{Simulations} \label{sec:simulations}
In this section, we explain how beam alignment can be performed in a wideband system. Then, we describe the system and channel parameters used in our simulations. Finally, we present numerical results that show the performance of CS-based beam alignment and zero filling-based beam alignment using FALP.
\subsection{Beam alignment in a wideband system} \label{sec:wbext}
\par  We use a Golay sequence-based frame structure \cite{11ad} and extend the beam alignment techniques in Sec.~\ref{sec:FALP_mainsubsection} and Sec.~\ref{sec:ZFB_procedure} to a wideband system. For an elaborate description of the wideband extension, we refer the reader to \cite{ZCglobecom}. We consider an $L$ tap wideband channel $\left\{ \mathbf{H}[\ell]\right\} _{\ell=0}^{L-1}$, where $\mathbf{H}[\ell] \in \mathbb{C}^{N \times N}$. For each phase shift configuration in $\{\mathbf{P}[m]\}^{M}_{n=1}$, the TX transmits a Golay complementary sequence of length $2N_{\mathrm{s}}$ followed by a guard interval of $L-1$ zeros. The use of guard interval prevents inter-frame interference and allows sufficient time to configure the phase shifters \cite{javiCS}. The RX uses the perfect autocorrelation property of complementary Golay sequences to obtain the channel impulse response (CIR) for each phase shift configuration. The CIR obtained in the $m^{\mathrm{th}}$ training slot is a noisy version of $\{ \langle \mathbf{H}[\ell], \mathbf{P}[m] \rangle \}_{\ell=0}^{L-1}$. Using several spatial channel projections, it is possible to reconstruct the wideband channel. We define $\mathbf{Y}_{\mathrm{blk}} \in \mathbb{C}^{M \times L}$ as a matrix that contains noisy wideband channel projections. The noise in $\mathbf{Y}_{\mathrm{blk}}$ is modelled using $ \mathbf{V}_s\in\mathbb{C}^{M \times L}$; $\mathbf{Y}_{\mathrm{blk}}(m,\ell)$ is then
\begin{equation}
\mathbf{Y}_{\mathrm{blk}}(m,\ell)= \langle \mathbf{H}[\ell], \mathbf{P}[m] \rangle + \mathbf{V}_s(m,\ell).
\end{equation}
As a spreading gain of $2 N_s$ is achieved at the output of the Golay correlator, it can be observed that the entries in $\mathbf{V}_s$ are independent and identically distributed as $\mathcal{N}_{\mathrm{c}}(0,\sigma^2/(2N_s))$. 
\par In this paper, a single tap of the wideband channel is used to determine the transmit beamformer. Nevertheless, CS-based wideband channel estimation can also be performed at the expense of a higher complexity \cite{cschest,javiCS}. The tap used to perform beam alignment is given by $\ell_{o}={\mathrm{argmax}}_{\ell}\Vert\mathbf{Y}_{\mathrm{blk}}(:,\ell)\Vert_2$. The channel measurements considered in FALP are compressive spatial projections of $\mathbf{H}[\ell_o]$, i.e., $\mathbf{y}=\mathbf{Y}_{\mathrm{blk}}(:,\ell_o)$. We use $\mathbf{H}[\ell_{\mathrm{opt}}]$ to denote the channel tap that has the maximum energy of the $L$ taps. In practice, $\mathbf{H}[\ell_{o}]$ can be different from $\mathbf{H}[\ell_{\mathrm{opt}}]$ as $\ell_{o}$ is determined from lower-dimensional spatial projections of the $L$ channel taps. The matrix $\hat{\mathbf{H}}[\ell_{o}]$ obtained using CS over $\mathbf{y}$, can be considered as an equivalent narrowband channel estimate that is used for beam alignment. Note that our approach ignores the correlation between channel taps as it performs beam alignment based on a single tap. Developing better beam alignment strategies that account for such correlations is an interesting direction.       
\subsection{System and channel description}
We consider an analog beamforming system in Fig.~\ref{fig:architect}, where the TX is equipped with a half-wavelength spaced UPA of size $32 \times 32$, i.e., $N=32$. The resolution of phase shifters is set to $q=1$ bit. We consider a carrier frequency of $28\, \mathrm{GHz}$ and an operating bandwidth of $100\, \mathrm{MHz}$, which corresponds to a symbol duration of $10\, \mathrm{ns}$. The height of the TX and the RX were $5\, \mathrm{m}$ and $2\, \mathrm{m}$ in our simulation setup. The separation between the TX and the RX is set to $60\, \mathrm{m}$. The transmit power at the TX is assumed to be $20\, \mathrm{dBm}$. The RX is equipped with a single antenna element.
\par The mmWave channels in our simulations were derived from the QuaDRiga channel simulator for the 3GPP 38.901 UMi-NLoS scenario \cite{Quadriga}. For a TX-RX separation of $60\,\mathrm{m}$, the omnidirectional RMS delay spread was found to be less than $176 \, \mathrm{ns}$ in more than $90 \%$ of the channel realizations. Considering the leakage effects due to pulse shaping, the wideband channel is modelled using $L=64$ taps corresponding to a duration of $640 \, \mathrm{ns}$. The channel measurements for CS-based beam alignment are acquired using Golay complementary sequences along the time dimension, where each sequence is of length $N_{\mathrm{s}}=64$. For our simulation settings, it can be observed that the duration of the guard interval that follows a Golay pair is $630\, \mathrm{ns}$. The guard interval is sufficient enough to change the phase shift configuration at the TX, as phase shifters with a settling time of about $30\, \mathrm{ns}$ at $28 \, \mathrm{GHz}$ have been reported in \cite{phaseshift}. The standard 2D-DFT is used as a sparsifying dictionary for the spatial channel, unless otherwise stated. The simulation results we report are the averages over $100$ channel realizations.
\subsection{Performance evaluation}
The 2D-CCS-based technique in Sec. \ref{sec:FALP_mainsubsection} estimates the beamspace $\mathbf{X}$, while ZFB in Sec. \ref{sec:ZFB_procedure} estimates a one-sparse approximation of $\mathbf{X}$. Therefore, we define different metrics and benchmarks to evaluate the two techniques. 
\begin{figure}[h!]
\centering
\includegraphics[trim=1.5cm 6.75cm 2.5cm 7.4cm,clip=true,width=0.47 \textwidth]{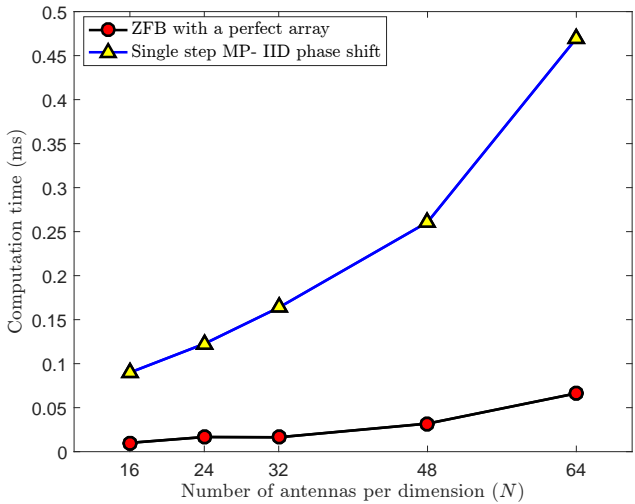}
  \caption{ZFB with FALP's training has a lower computational complexity than single step matching pursuit (MP) with random IID phase shift-based training \cite{anum}. Here, the number of measurements was set to $M=\mathrm{round}(20\, \mathrm{log}_{10}N^2)$. A similar complexity reduction can be observed with CS, due to the use of structured training in FALP.}
  \label{fig:Complexity_vs_ant}
\end{figure}
\par We describe the metrics and benchmarks used to evaluate 2D-CCS with FALP. In this paper, the error in the channel estimate is defined for a single tap, as CS is performed over the channel measurements corresponding to one tap. We use the OMP algorithm for CS-based channel estimation \cite{OMP_ref}. The stopping threshold and the maximum number of iterations used in OMP were $\sigma \sqrt{M/2N_{\mathrm{s}}}$ and $50$ \cite{OMP_ref}. The normalized squared error in the channel estimate is defined as 
\begin{equation}
\label{eq:NSE_defn}
\mathrm{NSE}= 20\, \mathrm{log}_{10}\left(\frac{\Vert \mathbf{H}[\ell_o]-\hat{\mathbf{H}}[\ell_o]\Vert_\mathrm{F}}{\Vert \mathbf{H}[\ell_o] \Vert_\mathrm{F}}\right).
\end{equation}
Using the channel estimate $\hat{\mathbf{H}}[\ell_o]$, the TX constructs the beamformer $\mathbf{F}_{\mathrm{CS}}$ according to the method in Sec.~\ref{sec:FALP_mainsubsection}. The number of elements in $\mathcal{B}$ was chosen as $K_{\mathcal{B}}=6$. The effective wideband single-input single-output (SISO) channel after CS-based beam alignment is then $ \{ \langle \mathbf{H}[\ell], \mathbf{F}_{\mathrm{CS}} \rangle \}^{L-1}_{\ell=0}$. The achievable rate is computed for the effective channel using the water filling algorithm. 
\par We compare the proposed CS-based approach with the perfect channel state information (CSI) scenario in which the beamformer is computed using $\mathbf{H}[\ell_{\mathrm{opt}}]$. We also evaluate random 2D-CCS, and standard CS with random IID phase shift matrices\cite{javiCS}. In random 2D-CCS, the base matrix $\mathbf{P}$ is chosen at random from the feasible set, i.e., $\mathbb{Q}^{N \times N}_q$ \cite{kramer}. It is important to note the transformation $\mathbf{S}=\mathbf{X}\odot \mathbf{Z}$ may not be invertible for a random base matrix. As a result, the approach in Fig.~\ref{fig:summary_FALP}, that solves for the beamspace through a partial 2D-DFT problem cannot be used for random 2D-CCS. The complexity of random 2D-CCS, however, is lower than standard CS. We use a low complexity version of OMP that exploits the 2D-convolutional structure of the training in random 2D-CCS.
\par Now, we define benchmarks for ZFB using FALP. As a one-sparse approximation of $\mathbf{X}$ is estimated in ZFB, its CS counterparts are those that estimate a single beamspace component. We consider the single step matching pursuit (MP) algorithm in \cite{anum} for two different training designs. The first design uses fewer 2D-circulant shifts of a random base matrix, while the second one is the common IID phase shift-based design \cite{javiCS}. We would like to highlight the fact that both ZFB and matching pursuit with 2D-circulant shifts of a random base matrix, exploit the 2D-FFT for fast estimation. The complexity of both these algorithms is lower than the one that uses the IID phase shift-based design, as shown in Fig.~\ref{fig:Complexity_vs_ant}. %The dimension of the antenna array was varied from $16 \times 16$ to $64\times 64$ in Fig.~\ref{fig:Complexity_vs_ant}. 
\begin{figure}[h!]
\centering
\includegraphics[trim=2cm 6.75cm 2.25cm 7.5cm,clip=true,width=0.47 \textwidth]{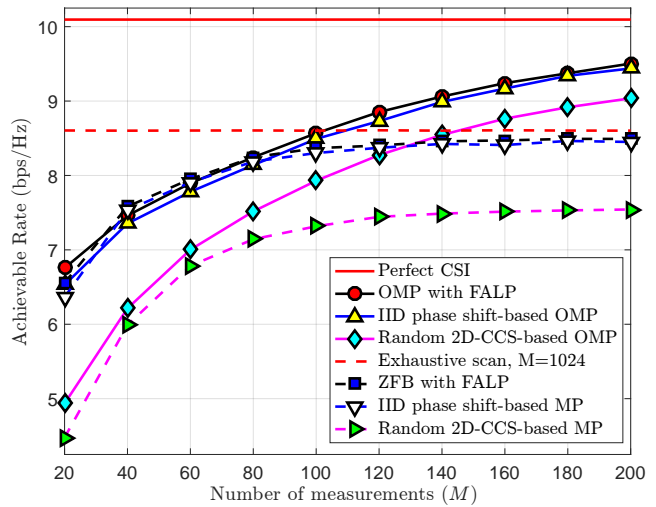}
  \caption{ Achievable rate with the number of CS measurements for a TX-RX separation of $60\, \mathrm{m}$. CS using FALP's training achieves better beam alignment than the random 2D-CCS-based training and the IID phase shift-based training. Similarly, ZFB with FALP's training results in a higher rate than MP with standard CS designs \cite{anum}. }
  \label{fig:Rate_vs_Meas}
\end{figure}
\par We compare the achievable rate obtained using the proposed techniques and the benchmarks. It can be observed from Fig.~\ref{fig:Rate_vs_Meas} that CS using the perfect array-based training in FALP achieves about $90 \%$ of the perfect CSI rate, with just $120$ channel measurements. In contrast, exhaustive scan-based beam alignment with the 2D-DFT dictionary requires $1024$ channel measurements. While both FALP and random 2D-CCS use 2D-convolutional channel acquisition, it can be observed that the rate achieved with FALP is significantly larger than that with random 2D-CCS. Similarly, ZFB performs better than single step MP that uses 2D-circulant shifts of a random base matrix. The difference in performance between the two techniques is due to the choice of the base matrix, i.e., $\mathbf{P}$. FALP uses a carefully designed base matrix, i.e., a perfect array, that satisfies the optimality conditions for efficient 2D-CCS. The loss in achievable rate when compared to the perfect CSI case is due to noise in the channel measurements and leakage effects in the beamspace representation. 
\begin{figure}[h!]
%\vspace{-3mm}
\centering
\includegraphics[trim=9.5cm 2cm 9.8cm 2.5cm,clip=true,width=0.47\textwidth]{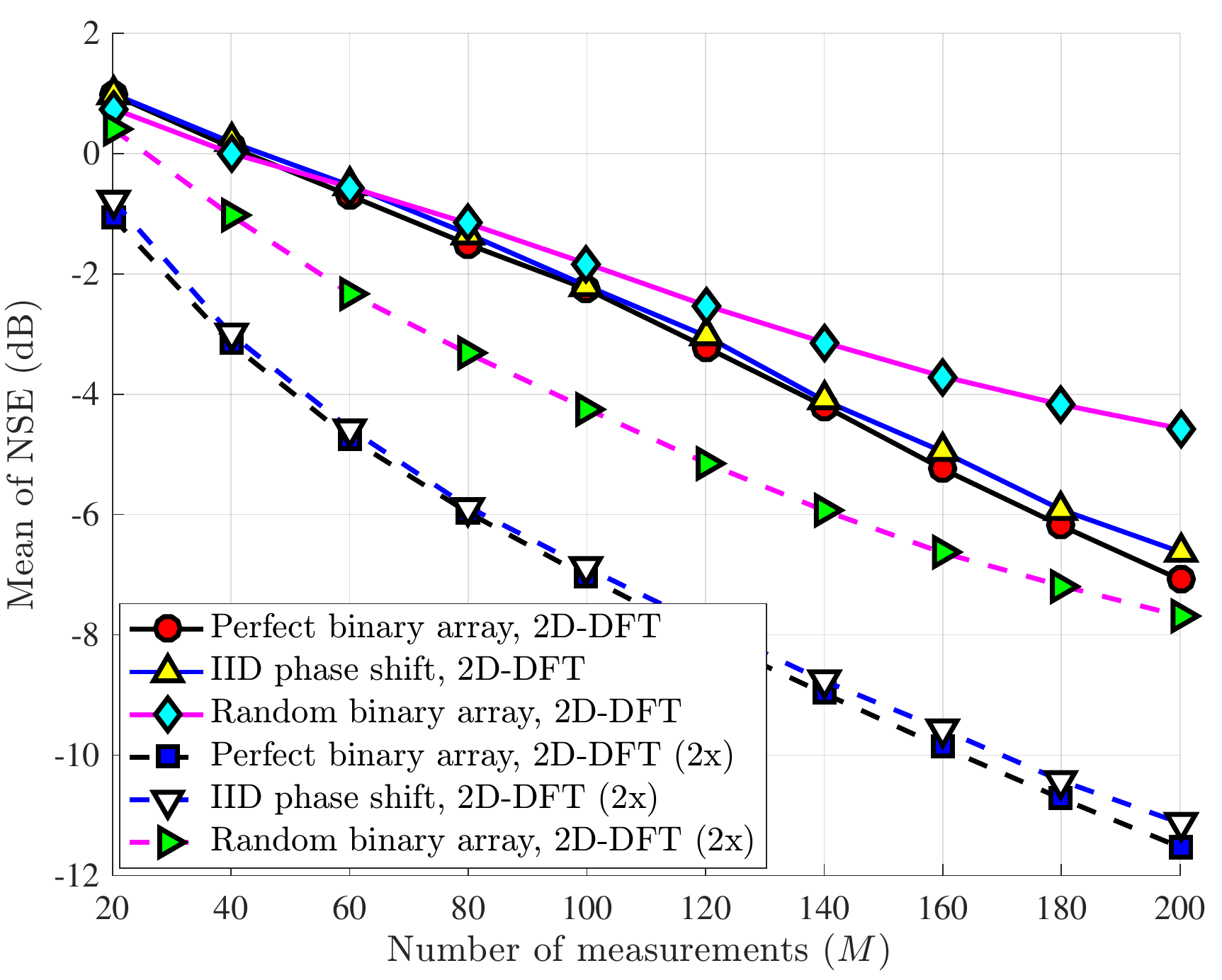}
  \caption{The plot shows the mean of the NSE in the channel estimate with the number of channel measurements. A $2\times$ resolution AoD domain representation of the beamspace results in a lower NSE. The complexity of CS algorithms that use such a representation, however, is higher than those that use the standard 2D-DFT dictionary. }
  \label{fig:MNSE_vs_SNR}
\end{figure}
\par The plot in Fig.~\ref{fig:Rate_vs_Meas} indicates that partial 2D-DFT CS-based beam alignment with FALP performs slightly better than standard CS with the IID random phase shift-based design. It is important to note, however, that the computational complexity of the CS algorithm in FALP is significantly lower than the standard approach that uses IID random phase shifts. As CS algorithms typically involve iterative optimization, it is useful to reduce the complexity of computations in each iteration without compromising the performance of the algorithm. The random 2D-CCS-based approach has a lower complexity, but results in a poor achievable rate. CS-based algorithms in  FALP achieve the best of both worlds, i.e., better beam alignment at a reduced computational complexity.
\par In Fig.~\ref{fig:MNSE_vs_SNR}, we plot the channel estimation accuracy with CS-based techniques. The metric used in Fig.~\ref{fig:MNSE_vs_SNR} is the mean of the normalized squared error in \eqref{eq:NSE_defn}, i.e., $\mathbb{E}[\mathrm{NSE}]$. We would like to point out that this metric is different from the usual normalized mean squared error (NMSE).  We use mean of NSE, as the normalized mean squared error (NMSE) is dominated by poor channel realizations that result in a low received power. As the primary objective of our simulations is to compare different training solutions, we propose to use $\mathbb{E}[\mathrm{NSE}]$, where the mean is taken over the NSE in dB. Such a metric is robust to fluctuations in the norm of the channel. The proposed CS-based beam alignment technique results in a slightly lower mean NSE than the random phase shift-based approach. Although the mean NSE approaches just $-7\, \mathrm{dB}$, it can be observed from Fig.~\ref{fig:Rate_vs_Meas} that the achievable rate with the proposed approach is reasonable for $M=200$. It is because beam alignment depends on how well the CS algorithm reconstructs the phase of the channel instead of the full channel.  From Fig.~\ref{fig:MNSE_vs_SNR}, it can be observed that the use of an oversampled 2D-DFT dictionary, by a factor of $2$ along both the azimuth and elevation dimensions, results in a lower mean NSE. In such a case, CS algorithms have a higher complexity as the dimensionality of the underlying optimization problem is quadrupled.
\begin{figure}[h!]
\centering
\includegraphics[trim=9.75cm 2cm 9.9cm 2.5cm,clip=true,width=0.47 \textwidth]{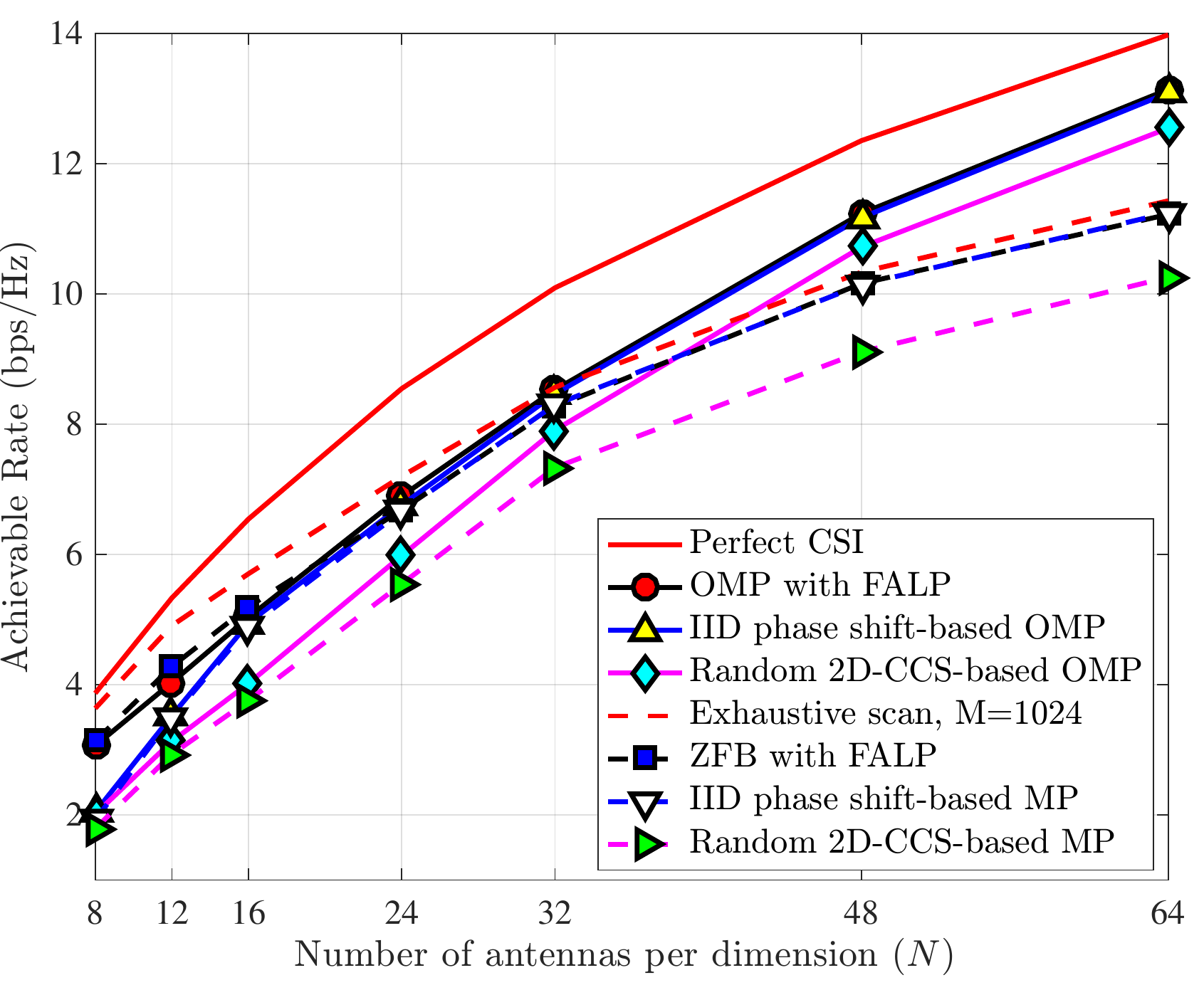}
  \caption{Achievable rate as a function of the number of antennas per dimension for $\rho=0.1$. A per-antenna power constraint of $-10\, \mathrm{dBm}$ was used for this simulation. Here, $N$ is chosen such that a perfect binary array of size $N \times N$ exists. FALP allows the use of low complexity CS algorithms, and is a promising solution for large arrays.}
  \label{fig:Rate_vs_Antt}
\end{figure}
\par The low complexity nature of CS algorithms in FALP makes it a promising solution for beam alignment in massive and low resolution phased arrays. In Fig.~\ref{fig:Rate_vs_Antt}, we show that CS-based beam alignment using FALP works well for a wide range of antenna dimensions. We would like to mention that FALP can be applied in one-bit phased arrays only when a perfect binary array exists, i.e., when $N=2^k$ or $N=3 \cdot 2^k$, for a natural number $k$. Designing perfect arrays for other combinations of phase shift resolution and antenna configurations, is an interesting research direction.
\section{Conclusions and future work}\label{sec:concl_fw}
In this paper, we have proposed FALP, a framework for compressive beam alignment or channel estimation using a perfect array-based codebook. The existence of perfect arrays over small alphabets allows fast and efficient compressed sensing in low resolution phased arrays through FALP. We have derived guarantees on channel reconstruction from sub-Nyquist sampling using FALP. In addition, we have shown how zero filling-based reconstruction in MRI can be used for rapid and low complexity beam alignment using a single 2D-FFT. 
\par FALP establishes a new platform to translate CS ideas from MRI to channel estimation or beam alignment in mmWave systems. As CS matrices in FALP can be parameterized by trajectories, k-space trajectories in MRI can be used for beam alignment or channel estimation. Furthermore, the trajectories can be optimized so that the CS matrix is robust to hardware non-idealities like CFO, phase noise, and frame synchronization errors. Investigating the performance of existing MRI trajectories in the beam alignment problem, and designing new trajectories are interesting research directions.
\section{Proof of Theorems} 
\subsection{Proof of Theorem \ref{theoremcondn}} \label{sec:proof_condn}
For the conditions in Theorem \ref{theoremcondn}, the reconstruction error in the masked beamspace obtained using the $\ell_1$-norm optimization program in \eqref{eq:l1opt} can be bounded as
\begin{equation}
\label{eq:msk_guarant}
\bigl\Vert\mathbf{S}-\hat{\mathbf{S}}\bigl\Vert_{F}\leq C_{1}\frac{\left\Vert \mathbf{S}-\left(\mathbf{S}\right)_{k}\right\Vert _{1}}{\sqrt{k}}+C_{2}N\sigma.
\end{equation}
The upper bound on the reconstruction error in \eqref{eq:msk_guarant} follows from Theorems 1 and 3 of \cite{kramer}. Using the spectral mask equation, i.e., $\mathbf{S}=\mathbf{X} \odot \mathbf{Z}$, and  \eqref{eq:msk_guarant}, we translate the guarantee on $\hat{\mathbf{S}}$ to the true beamspace estimate, i.e., $\hat{\mathbf{X}}$.
\par We first obtain an upper bound on the $\ell_1$ approximation error of the masked beamspace in \eqref{eq:msk_guarant}. We define $\Gamma \subseteq \mathcal{I}_N \times \mathcal{I}_N$ and its complement as $\Gamma ^{\mathrm{c}}$. The cardinality of $\Gamma$ is denoted by $|\Gamma|$. With the definition of $\mathbf{N}_{\Omega}$ in \eqref{eq:binary_mat}, we express the sparse approximation error $\Vert\mathbf{S}-(\mathbf{S})_{k}\Vert_{1}$ as
\begin{align}
\Vert\mathbf{S}-(\mathbf{S})_{k}\Vert_{1}&=\underset{\Gamma,|\Gamma|=k}{\mathrm{min}}\Vert\mathbf{S}-(\mathbf{S} \odot \mathbf{N}_{\Gamma})\Vert_{1}\\
\label{eq:l1_error_def}
&=\underset{\Gamma,|\Gamma|=k}{\mathrm{min}}\underset{(\ell,m)\in\Gamma^{\mathrm{c}}}{\sum}|\mathbf{S}(\ell,m)|.
\end{align}
From the spectral mask equation, we have $\mathbf{S}(\ell,m) = \mathbf{X}(\ell,m) \mathbf{Z}(\ell,m)$. As a result, $|\mathbf{S}(\ell,m)| \leq \Zmax |\mathbf{X}(\ell,m)|$. The $\ell_1$ approximation error in \eqref{eq:l1_error_def} is upper bounded as
\begin{align}
\label{eq:l1_error_def2}
\Vert\mathbf{S}-(\mathbf{S})_{k}\Vert_{1}& \leq \Zmax \underset{\Gamma,|\Gamma|=k}{\mathrm{min}}\underset{(\ell,m)\in\Gamma^{\mathrm{c}}}{\sum}|\mathbf{X}(\ell,m)|.
\end{align} 
Using the definition of the $\ell_1$ error in a $k$ sparse approximation of $\mathbf{X}$ in \eqref{eq:l1_error_def2}, we have $\Vert\mathbf{S}-(\mathbf{S})_{k}\Vert_{1} \leq \Zmax \Vert\mathbf{X}-(\mathbf{X})_{k}\Vert_{1}$.
\par Now, we derive a lower bound on the error in the masked beamspace estimate in terms of the error in the true beamspace estimate. The squared error in $\hat{\mathbf{S}}$ is  $\Vert\mathbf{S}-\hat{\mathbf{S}}\Vert_{F}^{2}=\sum_{\ell,m}|\mathbf{Z}(\ell,m)(\mathbf{X}(\ell,m)-\mathbf{\hat{X}}(\ell,m))|^{2}$. By definition, $|\mathbf{Z}(\ell,m)|\geq \Zmin$ for every $\ell$ and $m$. Therefore, the error in the masked beamspace estimate is lower bounded as $\Vert\mathbf{S}-\hat{\mathbf{S}}\Vert_{F} \geq \Zmin \Vert\mathbf{X}-\hat{\mathbf{X}}\Vert_{F}$. The result in Theorem \ref{theoremcondn} follows by using $\Vert\mathbf{S}-(\mathbf{S})_{k}\Vert_{1} \leq \Zmax \Vert\mathbf{X}-(\mathbf{X})_{k}\Vert_{1}$ and $\Vert\mathbf{S}-\hat{\mathbf{S}}\Vert_{F} \geq \Zmin \Vert\mathbf{X}-\hat{\mathbf{X}}\Vert_{F}$ in \eqref{eq:msk_guarant}.   
\subsection{Proof of Theorem \ref{theorem_array_duality}} \label{sec:proof_theorem_array_duality}
The spectral mask $\mathbf{Z}$ is a scaled inverse 2D-DFT of $\mathbf{P}_{\mathrm{FC}}$, i.e., $\mathbf{Z}=N\mathbf{U}^{\ast}_N \mathbf{P}_{\mathrm{FC}} \mathbf{U}^{\ast}_N$. It can be observed that the unimodularity of $\mathbf{Z}$ is equivalent to the unimodularity of $\mathbf{Z}_{\mathrm{FC}}$, a flipped and conjugated version of $\mathbf{Z}$. Using the properties of the Fourier transform, it can be shown that $\mathbf{Z}_{\mathrm{FC}}=N \mathbf{U}_N \mathbf{P}\mathbf{U}_N$ \cite{imageprocess}. Now, we use the result that $N \mathbf{U}_N \mathbf{P}\mathbf{U}_N$ is unimodular if and only if $\mathbf{P}$ has perfect periodic spatial autocorrelation \cite[Sec. III-A]{luke1988sequences}. The result in Theorem \ref{theorem_array_duality} follows by putting these observations together.
%\begin{theorem} 
%A matrix $\mathbf{P} \in \mathbb{Q}^{N \times N}_q$ has a unimodular spectral mask, i.e., $\Zmax=1$ and $\Zmin=1$, if and only if $\mathbf{P}$ has perfect spatial autocorrelation, i.e.,  
%\begin{equation}
%\label{eq:array_duality}
%\langle \mathbf{P}, \mathbf{J}_x \mathbf{P} \mathbf{J}^T_y \rangle =0\;\; \forall \, (x,y)\in \mathcal{I}_N \times \mathcal{I}_N \setminus (0,0).
%\end{equation}
%Expand definition of $\mathbf{P}_{\mathrm{FC}}$ and show this equivalence.
\subsection{Proof of Theorem \ref{theorem_prob_bound}} \label{sec:proof_prob_bound}
From \eqref{eq:pexact}, the probability that $|(\mathbf{S}\circledast\mathbf{K}_{\mathrm{bl}})|$ achieves maximum at $(0,0)$ can be expressed as 
\begin{equation}
\label{eq:pcomplement}
p=1-\mathrm{Pr}\bigg(\underset{(r,c)\neq(0,0)}{\cup}|(\mathbf{S}\circledast\mathbf{K}_{\mathrm{bl}})_{0,0}|\leq|(\mathbf{S}\circledast\mathbf{K}_{\mathrm{bl}})_{r,c}| \bigg).
\end{equation}  
Note that $\mathbf{S}$ is non-zero only at the locations $(0,0)$ and $(r_o,c_o)$. In this section, we derive closed form expressions for $\mathrm{Pr}(|(\mathbf{S}\circledast\mathbf{K}_{\mathrm{bl}})_{0,0}|\leq|(\mathbf{S}\circledast\mathbf{K}_{\mathrm{bl}})_{r_{o},c_{o}}|)$ and $\mathrm{Pr}(|(\mathbf{S}\circledast\mathbf{K}_{\mathrm{bl}})_{0,0}|\leq|(\mathbf{S}\circledast\mathbf{K}_{\mathrm{bl}})_{r_1,c_1}|)$ for some $(r_1,c_1) \notin \{(0,0),(r_o,c_o)\}$. We then derive a lower bound on $p$ in \eqref{eq:pcomplement}.      
\par The matrix $|\mathbf{S}|$ is $1$ at $(0,0)$, $|a|$ at $(r_o,c_o)$, and $0$ at the remaining $N^2-2$ locations. Although $|\mathbf{S}|$ is maximum at $(0,0)$ for $|a|<1$, it is possible that $|(\mathbf{S}\circledast\mathbf{K}_{\mathrm{bl}})|$ achieves maximum at $(r_o,c_o)$. We define $p_1=\mathrm{Pr}(|(\mathbf{S}\circledast\mathbf{K}_{\mathrm{bl}})_{0,0}|\leq|(\mathbf{S}\circledast\mathbf{K}_{\mathrm{bl}})_{r_o,c_o}|)$. Using \eqref{eq:conv_compact}, $p_1$ can be expressed as 
\begin{align}
p_1&=\mathrm{Pr}(|1+a\mathsf{x}^{\ast}| \leq |a+\mathsf{x}|).
\end{align}
The inequality $|1+a\mathsf{x}^{\ast}| \leq |a+\mathsf{x}|$ is equivalent to $1+|a|^2|\mathsf{x}|^2+2\mathrm{Re}\{a\mathsf{x}^{\ast}\} \leq |a|^2 +|\mathsf{x}|^2 +2\mathrm{Re}\{a\mathsf{x}^{\ast}\}$ and can be simplified to $(1-|a|^{2})|\mathsf{x}|^{2}\geq1-|a|^{2}$. For any $|a|<1$, $p_1$ is given by 
\begin{align}
p_1&=\mathrm{Pr}(|\mathsf{x}|^2 \geq 1)\\
\label{eq:firstprob_term_pre1}
&=\mathrm{Pr}(2|\mathsf{x}|^2/ \xi^2 \geq 2/\xi^2).
\end{align}
As $\mathsf{x}\sim \mathcal{N}_{\mathrm{c}}(0,\xi^2)$ for a large $N$, $2|\mathsf{x}|^2/\xi^2$ follows the central $\chi^2$-distribution of order $2$. Therefore, $p_1$ in \eqref{eq:firstprob_term_pre1} can be expressed in terms of the Marcum Q-function \cite{chi_sq} as 
\begin{equation}
\label{eq:firstprob_term}
p_1=\mathcal{Q}_1(0,\sqrt{2}/\xi).
\end{equation}
An interesting observation from \eqref{eq:firstprob_term} is that $p_1$ is independent of the strength of the second best path, i.e., $|a|$. 
\par Now, we derive the probability that $|(\mathbf{S}\circledast\mathbf{K}_{\mathrm{bl}})_{0,0}|\leq|(\mathbf{S}\circledast\mathbf{K}_{\mathrm{bl}})_{r_1,c_1}|$ for some $(r_1,c_1)$ such that $\mathbf{S}(r_1,c_1)=0$. We define $p_2=\mathrm{Pr}(|(\mathbf{S}\circledast\mathbf{K}_{\mathrm{bl}})_{0,0}|\leq|(\mathbf{S}\circledast\mathbf{K}_{\mathrm{bl}})_{r_1,c_1}|)$. Using \eqref{eq:conv_compact}, $p_2$ can be expressed as
\begin{equation}
\label{eq:p2_temp1}
p_2=\mathrm{Pr}(|1+a\mathsf{x}^{\ast}| \leq |\mathsf{b}+a\mathsf{w}|).
\end{equation}
We assume that the variables $\mathsf{x}$, $\mathsf{b}$ and $\mathsf{w}$ are independent. With this assumption, it can be observed that $\mathsf{b}+a\mathsf{w} \sim \mathcal{N}_{\mathrm{c}}(0,(1+|a|^2) \xi^2)$ and $1+a\mathsf{x}^{\ast} \sim \mathcal{N}_{\mathrm{c}}(1,|a|^2 \xi^2)$. We use $\chi^2_{_{\mathrm{C}}}$ and $\chi^2_{_{\mathrm{NC}}}$ to denote the central and non-central chi-squared random variables of degree $2$ \cite{chi_sq}. The non-centrality parameter of $\chi^2_{_{\mathrm{NC}}}$ is set as $\lambda_{_{\mathrm{NC}}}= 2/{|a|^2 \xi^2}$. Using these definitions, it can be shown that $|\mathsf{b}+a\mathsf{w}|^2 \sim \xi^2(1+|a|^2)\chi^2_{_{\mathrm{C}}}/2$ and $|1+a\mathsf{x}^{\ast}|^2 \sim |a|^2\xi^2 \chi^2_{_{\mathrm{NC}}}/2$. We use $f(t)$ to denote the probability density of $\chi^2_{_{\mathrm{NC}}}$  at $t$. The probability in \eqref{eq:p2_temp1} is then 
\begin{align}
p_2 &=\mathrm{Pr}\left(\chi^2_{_{\mathrm{C}}} \geq \frac{|a|^2}{1+|a|^2}\chi^2_{_{\mathrm{NC}}} \right)\\
\label{eq:p_2pre_ccdf}
&=\int_{0}^{\infty}\mathrm{Pr}\left(\chi_{_{\mathrm{C}}}^{2}\geq\frac{|a|^{2}}{1+|a|^{2}}t\right)f(t)dt.
\end{align} 
We use the complementary cumulative distribution function of $\chi^2_{_{\mathrm{C}}}$ to express \eqref{eq:p_2pre_ccdf} as 
\begin{align}
\label{eq:p_2pre_mgf}
p_2 &= \int_{0}^{\infty} \mathrm{exp}\left(\frac{-|a|^{2}t}{2(1+|a|^{2})}\right)f(t)dt.
\end{align}
It can be observed from \eqref{eq:p_2pre_mgf} that $p_2$ is the moment generating function of $\chi^2_{_{\mathrm{NC}}}$ \cite{chi_sq} evaluated at $-|a|^{2}/(2+2|a|^2)$. The probability in \eqref{eq:p_2pre_mgf} is then 
\begin{align}
\label{eq:prob_2}
p_2 &= \frac{1+|a|^{2}}{1+2|a|^{2}}\mathrm{exp}\left(\frac{-1}{(1+2|a|^{2})\xi^{2}}\right).
\end{align}
As a sanity check, it can be observed that $p_2=p_1$ for $a=0$.
\par We use $p_1$ and $p_2$ in \eqref{eq:firstprob_term} and \eqref{eq:prob_2}, to derive a lower bound on \eqref{eq:pcomplement}. The probability that $|\mathbf{S} \circledast \mathbf{K}_{\mathrm{bl}}|$ does not achieve its maximum at $(0,0)$ can be upper bounded using a union bound as
\begin{align}
\label{eq:1_p_temp1}
1-p & \leq \underset{(r,c)\neq(0,0)}{\sum}\mathrm{Pr}(|(\mathbf{S}\circledast\mathbf{K}_{\mathrm{bl}})_{0,0}|\leq|(\mathbf{S}\circledast\mathbf{K}_{\mathrm{bl}})_{r,c}|).
\end{align}
The right hand side in \eqref{eq:1_p_temp1} comprises of two distinct terms that are $p_1=\mathrm{Pr}(|(\mathbf{S}\circledast\mathbf{K}_{\mathrm{bl}})_{0,0}|\leq|(\mathbf{S}\circledast\mathbf{K}_{\mathrm{bl}})_{r_o,c_o}|)$, and $p_2=\mathrm{Pr}(|(\mathbf{S}\circledast\mathbf{K}_{\mathrm{bl}})_{0,0}|\leq|(\mathbf{S}\circledast\mathbf{K}_{\mathrm{bl}})_{r_1,c_1}|)$ for any $(r_1,c_1) \notin \{ (0,0),(r_{o},c_{o})\} $. As there are $N^2-2$ terms of the second kind, we can write 
\begin{align}
\label{eq:1_p_temp2}
1-p \leq p_1 + (N^2-2)p_2.
\end{align}
The result in Theorem \ref{theorem_prob_bound} follows by substituting \eqref{eq:firstprob_term} and \eqref{eq:prob_2} in \eqref{eq:1_p_temp2}.
\bibliographystyle{IEEEtran}
\bibliography{refs}
\begin{IEEEbiography}[{\includegraphics[width=1in,height=1.25in,clip,keepaspectratio]{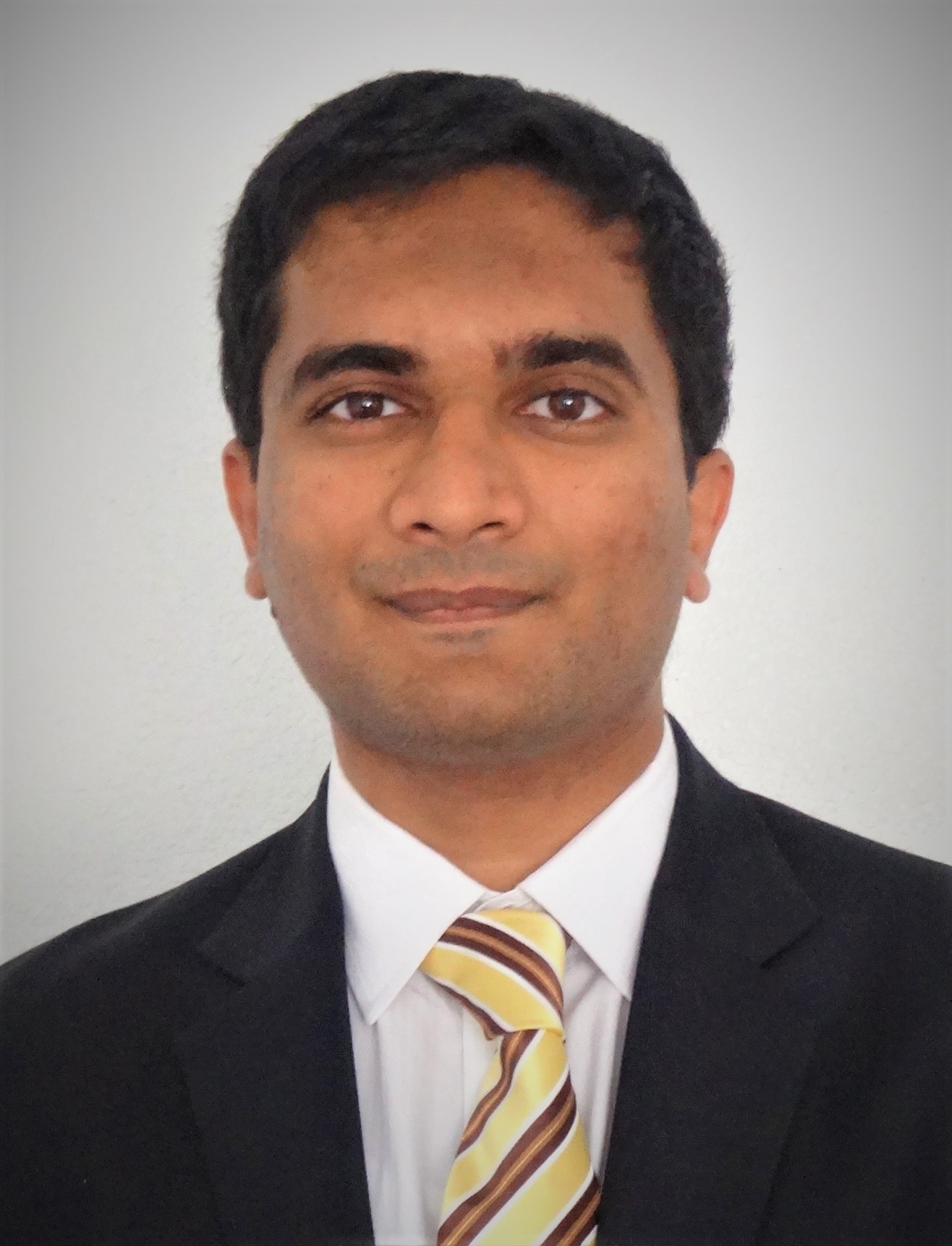}}]%
{Nitin Jonathan Myers} (S'17) received the B.Tech. and M.Tech. degrees in electrical engineering from the Indian Institute of Technology (IIT) Madras in 2016. He is currently pursuing the Ph.D. degree with The University of Texas at Austin. His research interests lie in the areas of wireless communications and signal processing. During his undergraduate days at IIT Madras, he received the DAAD WISE Scholarship, in 2014; and the Institute Silver Medal, in 2016. Mr. Myers is a recipient of the University Graduate Continuing Fellowship 2019-2020. He also received the 2018 and 2019 Electrical and Computer Engineering (ECE) Research Awards, and 2018 ECE Professional Development Award from the Cockrell School of Engineering, The University of Texas at Austin. 
\end{IEEEbiography}
\begin{IEEEbiography}[{\includegraphics[width=1in,height=1.25in,clip,keepaspectratio]{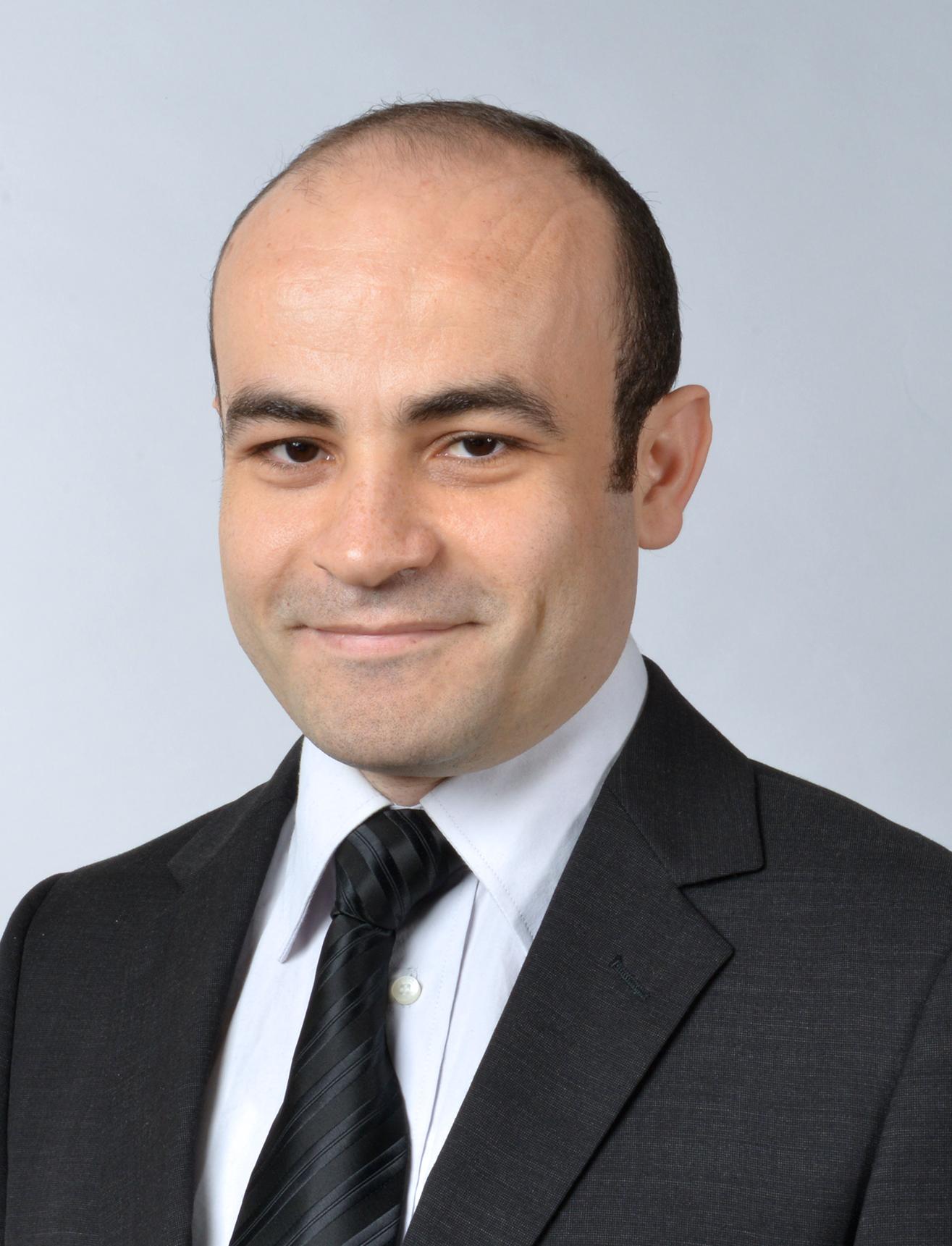}}]%
{Amine Mezghani} (S'08 - M'16) received the Ph.D. degree in Electrical Engineering from the Technical University of Munich, Germany in 2015. He is currently an Assistant Professor in the Department of Electrical and Computer Engineering at University of Manitoba, Canada. He was as a Postdoctoral Fellow at the University of Texas at Austin, USA, while contributing to this work. Prior to this, he was a Postdoctoral Scholar with the Department of Electrical Engineering and Computer Science, University of California, Irvine, USA. His current research interests include millimeter-wave communications, massive MIMO, hardware constrained radar and communication systems, antenna theory and large-scale signal processing algorithms. He was the recipient of the joint Rohde \& Schwarz and EE department Outstanding Dissertation Award in 2016.  He has published about hundred papers, particularly on the topic of signal processing and communications with low resolution analog-to-digital and digital-to-analog converters.
\end{IEEEbiography}
\begin{IEEEbiography}[{\includegraphics[width=1in,height=1.25in,clip]{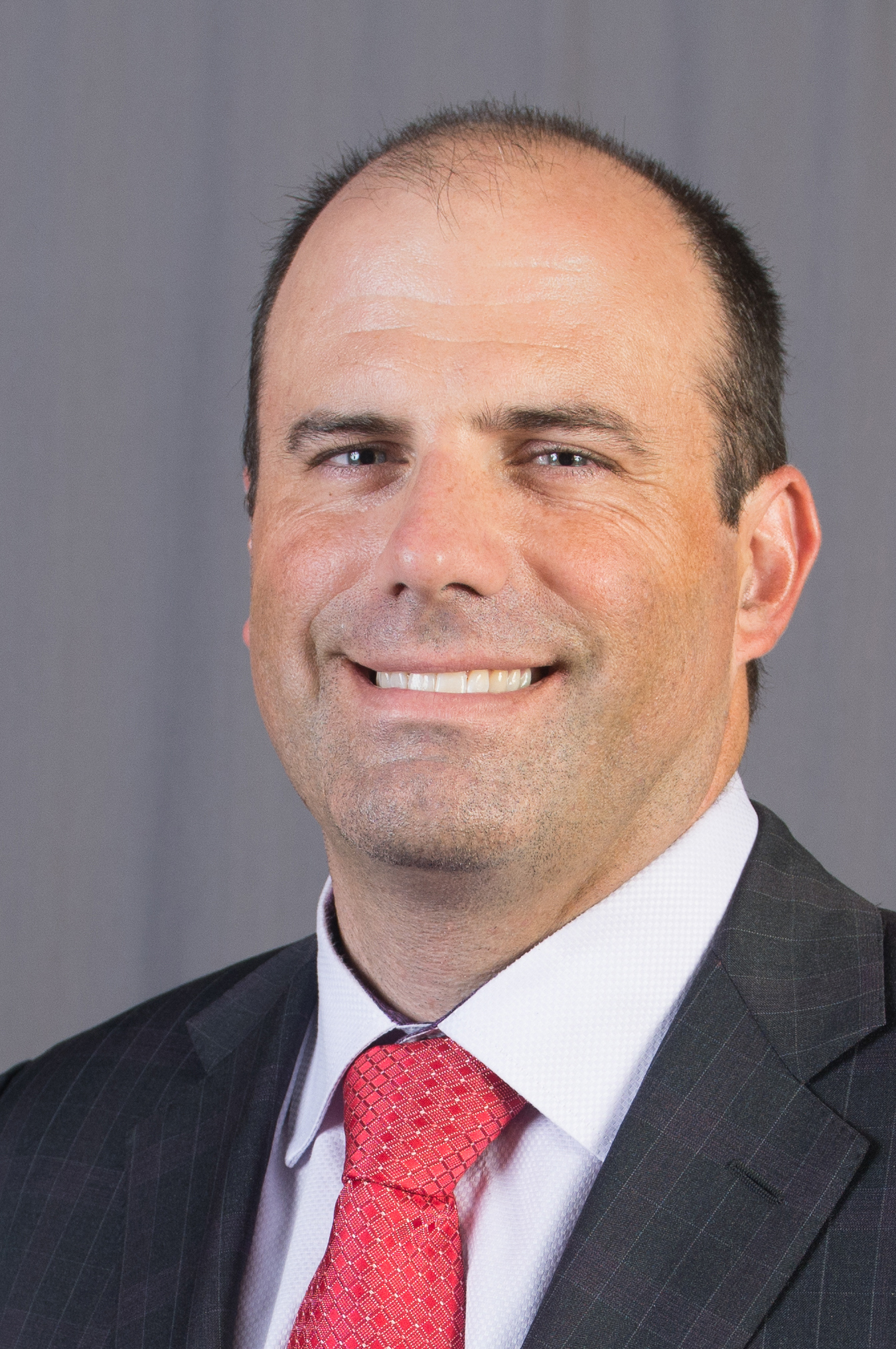}}]%
{Robert W. Heath Jr. } 
(S'96 - M'01 - SM'06 - F'11)  received the B.S. and M.S. degrees from the University of Virginia, Charlottesville, VA, in 1996 and 1997 respectively, and the Ph.D. from Stanford University, Stanford, CA, in 2002, all in electrical engineering. From 1998 to 2001, he was a Senior Member of the Technical Staff then a Senior Consultant at Iospan Wireless Inc, San Jose, CA where he worked on the design and implementation of the physical and link layers of the first commercial MIMO-OFDM communication system. Since January 2002, he has been with the Department of Electrical and Computer Engineering at The University of Texas at Austin where he is a Cockrell Family Regents Chair in Engineering, founder of the Situation-Aware Vehicular Engineering Systems initiative, and is a Member of the Wireless Networking and Communications Group. He is also President and CEO of MIMO Wireless Inc. He authored ``Introduction to Wireless Digital Communication’' (Prentice Hall, 2017) and ``Digital Wireless Communication: Physical Layer Exploration Lab Using the NI USRP’' (National Technology and Science Press, 2012), and co-authored ``Millimeter Wave Wireless Communications’' (Prentice Hall, 2014) and ``Foundations of MIMO Communication’' (Cambridge University Press, 2019).
\par Dr. Heath has been a co-author of a number award winning conference and journal papers including  the 2010 and 2013 EURASIP Journal on Wireless Communications and Networking best paper awards, the 2012 Signal Processing Magazine best paper award, a 2013 Signal Processing Society best paper award, 2014 EURASIP Journal on Advances in Signal Processing best paper award, the 2014 and 2017 Journal of Communications and Networks best paper awards, the 2016 IEEE Communications Society Fred W. Ellersick Prize, the 2016 IEEE Communications and  Information Theory Societies Joint Paper Award, the 2017 Marconi Prize Paper Award, and the 2019 IEEE Communications Society Stephen O. Rice Prize. He received the 2017 EURASIP Technical Achievement award  and is co-recipient of the 2019 IEEE Kiyo Tomiyasu Award. He was a distinguished lecturer and member of the Board of Governors in the IEEE Signal Processing Society. In 2017, he was selected as a Fellow of the National Academy of Inventors. He is also a licensed Amateur Radio Operator, a Private Pilot, a registered Professional Engineer in Texas. He is currently Editor-in-Chief of IEEE Signal Processing Magazine.
\end{IEEEbiography}
\end{document}